\documentclass{article}

\usepackage[nonatbib]{neurips_2021}
\usepackage{qcircuit}

\usepackage[utf8]{inputenc} 
\usepackage[T1]{fontenc}    
\usepackage{hyperref}       
\usepackage{url}            
\usepackage{booktabs}       
\usepackage{amsfonts}       
\usepackage{nicefrac}       
\usepackage{microtype}      
\usepackage{xcolor}    
\usepackage{graphicx} 

\newcommand{\Appendix}[1]{the full version for}

\newcommand{\p}{\bm{p}}

\newcommand{\x}{\bm{x}}
\newcommand{\y}{\bm{y}}

\newcommand{\I}{\bm{I}}

\newcommand{\X}{\bm{X}}
\newcommand{\Y}{\bm{Y}}

\usepackage{bm}
\usepackage{amsfonts}
\usepackage{amsmath}
\usepackage{amsthm}

\theoremstyle{definition}
\newtheorem{definition}{Definition}[section]
\newtheorem{theorem}[definition]{Theorem}
\newtheorem{corollary}[definition]{Corollary}
\newtheorem{proposition}[definition]{Proposition}

\usepackage{braket}
\usepackage{appendix}

\DeclareMathOperator{\Tr}{Tr}
\title{Are Quantum Circuits Better than Neural Networks at Learning Multi-dimensional Discrete Data? An Investigation into Practical Quantum Circuit Generative Models}

\author{%
  Pengyuan Zhai 
}

\begin{document}
\maketitle

\begin{abstract}
    Are multi-layer parameterized quantum circuits (MPQCs) more expressive than classical neural networks (NNs)? How, why, and in what aspects?  In this work, we survey and develop intuitive insights into the expressive power of MPQCs in relation to classical NNs. We organize available sources into a systematic proof on why MPQCs are able to generate probability distributions that cannot be efficiently simulated classically. We first show that instantaneous quantum polynomial circuits (IQPCs), are unlikely to be simulated classically to within a multiplicative error, and then show that MPQCs efficiently generalize IQPCs. We support the surveyed claims with numerical simulations: with the MPQC as the core architecture, we build different versions of quantum generative models to learn a given multi-dimensional, multi-modal discrete data distribution, and show their superior performances over a classical Generative Adversarial Network (GAN) equipped with the Gumbel Softmax for generating discrete data. In addition, we address practical issues such as how to efficiently train a quantum circuit with only limited samples, how to efficiently calculate the (quantum) gradient, and how to alleviate modal collapse. We propose and experimentally verify an efficient training-and-fine-tuning scheme for lowering the output noise and decreasing modal collapse. As an original contribution, we develop a novel loss function (MCR loss) inspired by an information-theoretical measure -- the coding rate reduction metric, which has a more expressive and geometrically meaningful latent space representations -- beneficial for both model selection and alleviating modal collapse. We derive the gradient of our MCR loss with respect to the circuit parameters under two settings: with the radial basis function (RBF) kernel and with a NN discriminator and conduct experiments to showcase its effectiveness.
\end{abstract}


\section{Introduction}
Whether the parameterized quantum circuits (PQCs) are more powerful than classical neural networks (NNs) in terms of expressiveness and data learning efficiency remains an open area of research. Showing that probability distributions generated by quantum circuits are hard to simulate classically is one crucial aspect of showcasing \textit{quantum supremacy} \cite{boixo_supremacy}, i.e., demonstrating that a programmable quantum device can solve a problem that no classical device can solve in any feasible amount of time. However, \textit{quantum supremacy} requires both theory and engineering-level guarantees. The main practical challenge is thus the high noise levels in near-term (noisy intermediate-scale quantum, NISQ) quantum devices. Therefore, demonstrating the (possible) expressive power advantage on NISQ devices becomes the current focus in the field.

Classical machine learning (ML) relies heavily on the expressive power of the respective model or network. Demonstrating that quantum devices/circuits are more expressive than classical NNs will also bring positive influence to the ML community. Model expressiveness is deeply rooted in the studies of discriminative learning and generative modeling. In this work, we investigate the expressive power of quantum circuits in the generative sense, i.e., how efficiently can a quantum circuit learn the underlying discrete probability distribution of some given observations (training data), and how much better is it than the classical NN, both theoretically and practically?

Generative Adversarial Networks (GANs) \cite{goodfellow2014generative} \cite{cgan} \cite{GAN-LOG} \cite{seqGAN} \cite{VAE-GAN} have been successful in classical generative learning tasks. However, GANs for discrete data have only been modeled heuristically and is an open area of research \cite{Gumbel_GAN}. Quantum systems, on the other hand, enjoy the naturally discrete structure and can be viewed as implicit discrete data generator, thanks to the superposition of states and probabilistic measurement outcomes (Born's rule \cite{Born_rule_german}).

We organize this paper as such: 
\begin{itemize}
    \item In Section \ref{related_work}, we list all relevant sources from two aspects: i. the theoretical proofs of the expressiveness and complexity of quantum circuits and ii. the related ML methods/practices; 
    \item In Section \ref{background_notations}, we review some background on quantum circuits (including Multilayer Parameterized Quantum Circuits (MPQCs) and Instantaneous Quantum Polynomial Time Circuits (IQPCs)) and classical computational complexity.  
    \item In Section \ref{proof_review}, from the surveyed works we organize a systematic proof on why MPQCs (and IQPCs) have superior expressive power than classical NNs, under a computational hardness assumption (the polynomial hierarchy does not collapses to its third level). 
    \item In Section \ref{ML_methods}, we go through relevant quantum and classical ML architectures, loss functions, and loss function gradient derivations, including the novel Maximal Coding Rate Reduction (MCR) metric and how it can be used as a sample-based distance measure between two distributions with a more expressive and geometrically meaningful feature space.
    \item in Section \ref{numerical_sims}, we show that, experimentally, MPQCs outperform classical NNs with the same number of parameters at learning the $2\times 2$ Bars and Stripes (BAS) distribution. We confirm that the statistical two-sample test relying on the Radial Basis Function (RBF) Kernel \cite{kernel_primer} with Maximal Mean Discrepancy (MMD) loss \cite{MMD_gretton} performs poorly with limited sample sizes. Meanwhile, we discover that adversarial training is effective in eliminating invalid image patterns, but suffers from serious modal collapse \cite{lala2018evaluation}. However, first training the quantum generative circuit via an adversarial minimax game against the discriminator and then fine-tuning the circuit parameters by the RBF Kernel against the real data samples works efficiently and effectively with small sample sizes. The MCR loss function also outperforms the traditional non-saturating GAN loss in terms of total variation and modal collapse.
    \item In Section \ref{conclusion}, we conclude and discuss future directions.
\end{itemize}

\section{Related Work}\label{related_work}
\paragraph{Instantaneous quantum polynomial circuits vs. NNs} Instantaneous quantum polynomial circuits (IQPCs) are a specific type of quantum circuits only composed of commuting gates in the $X$ basis. \cite{Bremner_PH} studied the simulability of IPQCs by classical processes and proved that IQPCs can output probability distributions that cannot be efficiently simulated classically.

\paragraph{Expressiveness of Multi-layer
Parameterized Quantum Circuits} Building on top of \cite{Bremner_PH}, authors of \cite{expressive_q} further studied the expressiveness of multi-layer
parameterized quantum circuits (MPQCs). They proved that MPQCs can efficiently generalize IQPCs, namely, each arbitrary circuit component in IQPC can be represented by constant numbers of MPQC blocks. In this work, we will organize the works of \cite{Bremner_PH} and \cite{expressive_q} and develop insights into the theorems with intuitive examples and redo some of the derivations with enhanced clarity.

\paragraph{Generative models} Generative machine learning models (\cite{GAN-LOG}\cite{cgan}\cite{infogan}\cite{Nash-Gan}) have enjoyed much success in recent years. Most noticeably, Generative Adversarial Networks (GANs) have been widely used for generating artificial images, text-to-image generation or image augmentation across areas of science, arts and media. However, besides issues such as being computationally expensive, convergence failure and modal collapse \cite{veegan_mode_collapse} \cite{che2016mode}, discrete-data learning GANs remain heuristically driven and an open area of research (\cite{Fedus2018MaskGANBT},\cite{seqGAN}).

\paragraph{Generative variational quantum circuit} Variational quantum circuits have been explored as a data generator in many studies (\cite{ref14dif}, \cite{ref15dif}, \cite{benedetti-quantum-shallow}). The ``quantum Born machine'' (\cite{dif_born}) exactly uses the inherent probabilistic interpretation of quantum states to output discrete data samples via projective measurement on the qubits. By using a variational quantum circuit parameterized by a set of rotation angles, \textbf{the circuit parameters can be tuned} to manipulate the probability distribution implicitly defined by the output quantum state, which leads to a \textbf{trainable} quantum circuit generator as used in (\cite{dif_born}, \cite{Stein2020QuGANAG}).

\paragraph{Coding Rate Reduction.} Recently, \cite{ReduNet} proposed a novel objective for learning a linear discriminative representation (LDR) for multi-class data, called ``coding rate reduction''. Unlike conventional discriminative methods only trying to predict class labels, the LDR learns the structure of the data, and is more suitable for both discriminative and generative purposes. \cite{ReduNet} maps distributions of the input data on {\em multiple} nonlinear submanifolds to multiple distinctive linear subspaces. Intuitively, this approach generalizes nonlinear PCA \cite{Kramer1991NonlinearPC} to settings where one simultaneously applies \em{multiple nonlinear PCAs} to data on multiple nonlinear submanifolds.

\section{Background and Notations}\label{background_notations}
\subsection{Quantum system and Born's rule} The structure of quantum systems naturally supports the learning of discrete data because, for a given qubit with state $\ket{\psi}=\alpha \ket{0}+\beta\ket{1}$, the state collapses to either $\ket{0}$ or $\ket{1}$ with probabilities $|\alpha|^2$ and $|\beta|^2$ upon measurement, which is a discrete distribution. With increasing numbers of qubits, the expressive power of modeling a discrete distribution increases exponentially, as for $n$ qubits the state space grows to $\ket{b}_1\otimes\ket{b}_2\otimes\dots \otimes \ket{b}_n, \ket{b}_i\in \{\ket{0}, \ket{1}\}$. The probability distribution on this discrete support is implicitly modeled by the $n$-qubit superpositioned quantum state $\ket{\psi}_n = \sum_{x=0}^{2^n-1} \alpha_x \ket{x}$, $\alpha_x\in \mathbb{C}$, where $|\alpha_x|^2$ is the probability that upon measurement, $\ket{\psi}_n$ will collapse to $\ket{x}$ (by Born's rule discussed below), where $x$ is a bit string representing (in base-two) the number between $0$ and $2^n-1$ (for example, $x=000...00, \; 000...01, \; 000...10, \; \dots, \; 111...11$). These $2^n$ bit strings form a orthonormal basis, called \emph{the standard computational basis}. 


The \textbf{Born's rule} \cite{Born_rule_german} is an important postulate of quantum mechanics which gives the probability that a measurement of a quantum system will yield a given result. For the simplest example, one measures the outputs state in the computational basis and produce a classical sample (bit string) by the following probability: $x \sim p_{\theta}(x)=\left|\left\langle x \mid \psi_{\theta}\right\rangle\right|^{2}$.

\subsection{Quantum circuit}
A quantum circuit applies quantum gates to quantum bits (or qubits). Mathematically, quantum gates are unitary matrix operations on a given quantum state vector. The set of single-qubit rotations, phase shift, and the ($2$-qubit) Controlled-NOT (CNOT) gates, can universally compute any function. The phase rotation gate $R_{\phi}$, z-axis rotation gate $R_{Z}(\theta)$, x-axis rotation gate $R_{X}(\gamma)$, and y-axis rotation gate $R_{Y}(\alpha)$ are as follows (where $\phi, \theta, \gamma, \alpha$ are rotation angles):
\begin{equation}
    \begin{gathered}
R_{\phi}=\left(\begin{array}{cc}
1 & 0 \\
0 & e^{i \phi}
\end{array}\right), \quad R_{X}(\gamma)=\left(\begin{array}{cc}
\cos (\gamma / 2) & i \sin (\gamma / 2) \\
i \sin (\gamma / 2) & \cos (\gamma / 2)
\end{array}\right) \\
R_{Z}(\theta)=\left(\begin{array}{cc}
e^{i \theta / 2} & 0 \\
0 & e^{-i \theta / 2}
\end{array}\right), ~~ R_{Y}(\alpha)=\left(\begin{array}{cc}
\cos (\alpha / 2) & \sin (\alpha / 2) \\
-\sin (\alpha / 2) & \cos (\alpha / 2)
\end{array}\right).
\end{gathered}
\end{equation}

Additionally, the CNOT gate flips the target qubit iff the the control qubit is $1$. This creates entanglement between any arbitrary pair of qubits:
$
\mathrm{CNOT}=\left(\begin{array}{llll}
1 & 0 & 0 & 0 \\
0 & 1 & 0 & 0 \\
0 & 0 & 0 & 1 \\
0 & 0 & 1 & 0
\end{array}\right). 
$

Other than the gates shown above, there are many other \textit{universal gat sets}. In this work, we will in particular discuss the set composed of {$H, Z, C Z \text { and } T=R_{\phi}(\phi=\pi / 4)$} which will appear in any instantaneous quantum polynomial (IQP) circuit, because other than the $H$ gate, all other gates are diagonal in the $Z$ basis, and we can thus append an $H$ or $H\otimes H$ (depending on $1$-qubit or $2$-qubit operations) gate before and after these $Z$-diagonal gates to produce gates that are diagonal in the $X$ basis as required by IQP circuits. Note $\mathrm{CZ}=\left(\begin{array}{llll}
1 & 0 & 0 & 0 \\
0 & 1 & 0 & 0 \\
0 & 0 & 1 & 0 \\
0 & 0 & 0 & -1
\end{array}\right)$.

\subsection{Computational Tasks} Classically or quantumly, a computational task involves the input, the process/circuit, and the output (often probabilistic in nature). These three steps elegantly summarize the basic ML tasks: for discriminative tasks, the input is data, the process is the discriminative model (such as NNs), and the output is some probabilistic prediction; for generative models, the input is some latent code, the process is the generative model, and the output is synthesized data.

In this study we use ``process'' and ``circuit'' interchangeably. How does one measure the ``closeness'' or similarity between the outputs from two computational processes? The answer is through the concept of \textit{error}:
\begin{definition}[Bounded error]
Given two computational processes $\mathcal{T}$ and $\mathcal{C}$, let the input be $w$, which is some bit string, and the outputs be $\mathcal{T}(w)$ and $\mathcal{C}(w)$ of the same arbitrary dimension. $\mathcal{C}$ computes $\mathcal{T}$ with \textit{bounded error} if there is a constant error $0<\epsilon<\frac{1}{2}$ such that for all inputs $w\in W$, $ \text{Prob}[\mathcal{C}(w)=\mathcal{T}(w)]\geq 1-\epsilon$. 
\end{definition}

\begin{definition}[Unbounded error]
Given two computational processes $\mathcal{T}$ and $\mathcal{C}$ as defined above. $\mathcal{C}$ computes $\mathcal{T}$ with \textit{unbounded error} for all inputs $w\in W$, $\text{Prob}[\mathcal{C}(w)=\mathcal{T}(w)]> \frac{1}{2}$. 
\end{definition}

In addition, an important computational ask is the \textit{decision task}, in which a circuit $\mathcal{C}$ only outputs a single bit. In ML applications, the binary classification task belongs to this scenario.

\subsection{Probability distribution of circuit output}
Suppose we are given a (classical or quantum) circuit $\mathcal{C}$, which takes a length-$m$ bit string $w$ as input. Associated to $\mathcal{C}$, there is a probability distribution $P_w$ on $m$-bit binary strings (where $m$ is the output dimension of $\mathcal{C}$), which models the output $\mathcal{C}(w)\sim P_w$, fixing an arbitrary input $w$. Note that this probabilistic model covers the case of deterministic outputs (
where $P_w$ puts probability $1$ on some point). For quantum circuits, the probability distribution of the circuit output is characterized by Born's rule \cite{Born_rule_german} of quantum measurement. The input is the standard state $\ket{0}...\ket{0}$ and the output is the probabilistic measurement result on a designated register of output lines in the computational basis. Specifically, the input state $|z\rangle=|0\rangle^{\otimes N}$ is evolved via some unitary operations $\left|\psi_{\theta}\right\rangle=U_{\theta}|z\rangle$. One measures the output state in the computational basis to produce a classical sample (bit string) $x \sim p_{\theta}(x)=\left|\left\langle x \mid \psi_{\theta}\right\rangle\right|^{2}$. This principle is the core inspiration behind the quantum Born Machine \cite{benedetti-quantum-shallow}, which our generative models are based on in this study.

\subsection{Complexity classes}
A complexity class is defined as a type of computational problem with a bounded resource like time or memory. In this study we are primarily concerned with the following complexity classes (all for decision tasks):
\begin{itemize}
    \item BPP, or bounded-error probabilistic polynomial time, describes the classical randomized polynomial time computation with bounded error, i.e., a classical circuit $\mathcal{C}$ that has a designated input portion and takes randomized bits at each run of the computation. 
    \item PP, or probabilistic polynomial time, describes the classical randomized polynomial time computation with unbounded error for the same classical randomized circuit described above.
    \item BQP, or bounded-error quantum polynomial time, describes the quantum polynomial time computation acting on a constant number of lines with bounded error.
\end{itemize}
We made a simple chart (Fig.~\ref{fig:complexity_chart}) to help with visualization.
\begin{figure}
    \centering
    \includegraphics[scale=0.3]{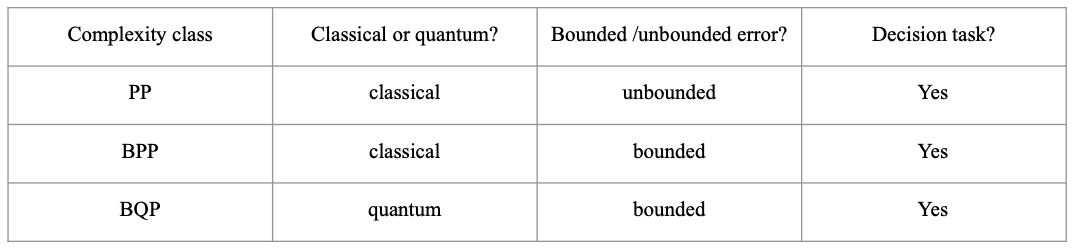}
    \caption{The complexity classes relevant to this study}
    \label{fig:complexity_chart}
\end{figure}

\subsection{Polynomial hierarchy}
In the case when a complexity class $A$ has oracle access to another complexity class $B$, we denote the new class as $A^B$. Basing on this notion, the polynomial hierarchy class (PH) \cite{boaz_complexity} is the union of an infinite series of increasing classes $\Delta_k, k=1,2,...$, where $\Delta_1=\mathrm{P}$, $\Delta_{k+1}=\mathrm{P}^{\mathrm{N} \Delta_k}$, $\mathrm{N} \Delta_{k}$ denotes the non-deterministic class associated to $\Delta_{k}$, in the same way that NP denotes the non-deterministic class associated to $\mathrm{P}$.
It is known that $\mathrm{P} = \mathrm{NP}$ iff the polynomial hierarchy collapses to its first level ($\mathrm{PH} = \mathrm{P}$) \cite{boaz_complexity}, which is unlikely to happen.
We will be concerned with a weaker yet similarly unlikely event: $\mathrm{PH}$ collapses to its third level ($\mathrm{PH} = \Delta_3$).
According to \cite{Bremner_PH}, we have the following fact
\begin{equation}
    \mathrm{P}^{\text {post }-\mathrm{BPP}} \subseteq \mathrm{P}^{\Delta_{3}}=\Delta_{3},
\end{equation}
and by Toda's Theorem \cite{boaz_complexity}:
\begin{equation}
    \mathrm{PH} \subseteq \mathrm{P}^{\mathrm{PP}}. 
\end{equation}

\subsection{Multilayer Parameterized Quantum (MPQ) Circuits}
MPQ circuits, or MPQCs, are multi-layer circuits, where each layer (or block) is composed of single-bit rotation unitaries and entanglement operations through CNOTs. The same block architecture is repeated $L$ times where $L$ is the total number of layer/depth of the circuit. Concretely, at the $l$-th layer, the rotation operation for qubit $j$ is
\[U\left(\theta_{l}^{j}\right)=R_{z}\left(\theta_{l}^{j, 1}\right) R_{y}\left(\theta_{l}^{j, 2}\right)R_{z}\left(\theta_{l}^{j, 3}\right) R_{\phi}\left(\phi \right) R_{x}\left(\theta_{l}^{j, 4}\right) R_{z}\left(\theta_{l}^{j, 5}\right) R_{x}\left(\theta_{l}^{j, 6}\right),\]
parameterized by the rotation angles $\{\theta_l^{j,1}, \theta_l^{j,2}, \theta_l^{j,3}, \theta_l^{j,4}, \theta_l^{j,5}, \theta_l^{j,6}, \phi_l^{j}\}$, where $R_{m}(\theta) \equiv \exp \left(\frac{-i \theta \sigma_{m}}{2}\right)$ is the unitary for rotating around axis $m\in \{z, x, y\}$ at an angle of $\theta$, and $R_{\phi}(\phi)$ is the phase rotation gate by an angle of $\phi$. The entanglement through CNOT does not require parameters, and different studies have explored different connectivities (which qubit is connected to which other qubit through CNOT) besides pair-wise full entanglement \cite{dif_born}. We will not overly discuss the connectivity choices here.

\subsection{Instantaneous Quantum Polynomial Time (IQP) Circuits}
IQP circuits, or IQPCs, are quantum circuits containing only commuting gates (in the $X=\{\frac{\ket{0}\pm \ket{1}}{\sqrt{2}}\}$ basis). The input is the standard $
\ket{0}\ket{0}\dots\ket{0}$ state, and the output is based on the measurement results in the computational basis.

Due to the commuting property, the gates in an IQPC can be applied simultaneously to respective qubits. Conventionally, previous works also represent IQPCs in terms of the Hadamard ($H$) gate 
and gates diagonal in the $Z$ basis (the standard computational basis): define the original IQPC as $U_\theta \ket{00\dots0}$, where $U_\theta=U_1 U_2\dots U_M$ is composed of multiple diagonal gates in the $X$ basis acting on all $N$ qubits. Then, since $HH = I$, we have
$U_\theta \ket{00\dots0}=H^{\otimes N} H^{\otimes N} U_\theta H^{\otimes N} H^{\otimes N} \ket{00\dots0}=H^{\otimes N}\underbrace{H^{\otimes N} U_1 H^{\otimes N}}_{V_1}\underbrace{H^{\otimes N} U_2 H^{\otimes N}}_{V_2}\dots \underbrace{H^{\otimes N} U_M H^{\otimes N}}_{V_M} H^{\otimes N} \ket{00\dots0}$,
where $V_1,\dots V_N$ are gates diagonal in the $Z$ basis.

\subsection{Post-selected IQP Complexity Class}
An important analytical tool used by \cite{Bremner_PH} is (classical or quantum) \emph{post-selected} circuits.  In addition to a specified register of output lines $\mathcal{O}$, a post-selected circuit has a further specified register of post-selection lines $\mathcal{P}$, disjoint with $\mathcal{O}$.  The sampling procedure of a post-selected circuit is described by the conditional distribution of $\mathcal{O}$ given $\mathcal{P} = 00...0$:
\begin{equation}
    \operatorname{Prob}[\mathcal{O}=x \mid \mathcal{P}=00 \ldots 0]=\frac{\operatorname{Prob}[\mathcal{O}=x \, \&\, \mathcal{P}=00 \ldots 0]}{\operatorname{Prob}[\mathcal{P}=00 \ldots 0]},
\end{equation}
as opposed to sampling from the un-conditional distribution $\operatorname{Prob}[\mathcal{O}=x]=\left|\left\langle x \mid \psi_{\theta}\right\rangle\right|^{2}$.
The conditional sampling is done classically, by simply counting the number of occurrences of outcome $\mathcal{O}=x$ among all observations where $\mathcal{P}=00\dots0$. We assume here that $\operatorname{Prob}[\mathcal{P}=00 \ldots 0]\neq 0$. Please note that this is a purely analytical tool, as the actual sampling procedure might be very inefficient if $\operatorname{Prob}[\mathcal{P}=00 \ldots 0]\approx 0$. The authors in \cite{Bremner_PH} established post-selected Instantaneous Quantum Polynomial Time (post-IQP) as a complexity class, along with post-selected BQP and post-selected BPP by the following definition:
\begin{definition}[\cite{Bremner_PH}]
\label{def: l_post_IQP}
A  language $L$ (a subset of bit strings of length $m$) is in the class post-IQP (resp. post-BQP or post-BPP) iff there is an error tolerance $0<\epsilon<\frac{1}{2}$ and a family $\{\mathcal{C}_w\}$ of post-selected IQP (resp. quantum or randomized classical) circuits\footnote{\cite{Bremner_PH} further requires $\{\mathcal C_w\}$ to be a ``uniform'' family, which means that the mapping from $w$ to the description of the circuit $\mathcal C_w$ can be classically computed in polynomial time.  We omit the term ``uniform'' for simplicity. } with a specified single line output register $\mathcal{O}_{w}$ 
and a specified (generally $O(\text{poly}(n))$-line) post-selection register $\mathcal{P}_{w}$ such that:
(i) if $w \in L$ then $\operatorname{Prob}\left[\mathcal{O}_{w}=1 \mid \mathcal{P}_{w}=00 \ldots 0\right] \geq 1-\epsilon$ and
(ii) if $w \notin L$ then $\operatorname{Prob}\left[\mathcal{O}_{w}=0 \mid \mathcal{P}_{w}=00 \ldots 0\right] \geq 1-\epsilon$.
\end{definition}

\subsection{Classical Simulability}
Given any process $\mathcal{T}$ (quantum or classical), how easy/hard is it to simulate this output probabilities of $\mathcal{T}$ by another classical process? \textit{Simulability} is a notion about the efficiency and accuracy of classical simulations. Given a circuit $\mathcal{C}_w$ that takes $w$ as input (a length-$n$ bitstring) and $P_w$ is the output probability distribution:
\begin{itemize}
    \item Strong simulability: a circuit is strongly simulable if any output probability in $P_w$ and any marginal probability of $P_w$ can be computed to $m$ digits of precision in classical poly$(n, m)$ time.
    \item $\mathcal{C}_w$ is weakly simulable if its output distribution $P_w$ can be \textit{sampled} by purely classical means in poly$(n)$ time.
\end{itemize}
Then, a circuit is weakly simulable with multiplicative error $c \geq 1$ if there is a distribution $R_w$ (on the same sample space as $P_w$) such that $R_w$ can be sampled in classical poly$(n)$ time and for all $x$ and $w$ we have
    \begin{equation}
        \frac{1}{c} \operatorname{Prob}\left[P_{w}=x\right] \leq \operatorname{Prob}\left[R_{w}=x\right] \leq c \operatorname{Prob}\left[P_{w}=x\right]
    \end{equation}

\section{Proof Review: MPQCs are more expressive than classical NNs}\label{proof_review}
\paragraph{Overview} The proof is organized based on proofs from two surveyed works, \cite{Bremner_PH} and \cite{expressive_q} with the following logic: the output probability distributions generated by IQPCs cannot be simulated to within multiplicative error classically, under a classical computational hardness assumption (Part 1); the Instantaneous Quantum Polynomial Time Circuits (IQPC) can be efficiently substituted by the  Multi-layer
Parameterized Quantum Circuits (MPQCs) (Part 2). Therefore, MPQCs can generate probability distributions that cannot be efficiently simulated classically, and hence is more expressive than classical processes (including NNs). \textbf{Rather than merely repeating the original proofs, we also give intuitive examples to enhance clarity}.

We first define relative ``expressiveness'' as follows: if process $A$ can output a probability distribution that $B$ cannot     ``efficiently'' simulate/sample, we say that $A$ is more expressive than $B$, where we define ``efficiency'' as ``weakly simulable to within some \textit{multiplicative error}'':

\subsection{Part 1: IQP circuits cannot be efficiently simulated by classical circuits, unless $\mathrm{PH} = \Delta_3$}
\begin{theorem}[\cite{Bremner_PH}] \label{thm:post-IQP=PP}
\emph{post-IQP = post-BQP = PP.} 
\end{theorem}
\begin{proof}
It has been proved by \cite{aaronson_quantum_2005} that post-BQP = PP.  So we only need to prove post-IQP = post-BQP. 
First note that $\text{post-IQP} \subseteq \text{post-BQP}$, because BQP does not put any restriction on the quantum circuit's architecture. Thus we prove the reverse inclusion, $\text{post-BQP} \subseteq \text{post-IQP}$. This is done by equivalently transforming the circuit architecture. Without loss of generality, assume that we use the following universal set of gates $H, Z, C Z \text { and } P=e^{i \frac{\pi}{8} Z}$ (except for $H$ they are all diagonal in the $Z$ basis). We then show that a post-selected BQP circuit using this universal gate set can always be reduced to post-IQP with the same output probability distribution as the original post-BQP circuit.

We use the ``Hadamard gadget'' trick, as illustrated in Fig \ref{fig:H_gadget}. Using the $H$ gate in the original circuit as an incision point, one can teleport the state of the qubit on the top line to the bottom line via a $CZ$ and two $H$ gates. Because $HH=I$, for each line we can add $HH$ in arbitrary positions, and substitute one of the $H$'s by the ``Hadamard gadget''. 
\textbf{We hereby show an intuitive example} of a post-selected BQP circuit with three qubits and three unitaries $U_1$, $U_2$, and $U_3$, which are composed of the gates from the given universal set ($H, Z, C Z \text { and } P=e^{i \frac{\pi}{8} Z}$). The last line is post-selected in the end:\\

\hspace{5em} 
\Qcircuit @C=1em @R=.7em {
  \lstick{\ket{0}} & \multigate{2}{U_1} & \multigate{2}{U_2}  & \multigate{2}{U_3} & \meter     \qw \\
  \lstick{\ket{0}} & \ghost{U_1} & \ghost{U_2} & \ghost{U_3} & \meter    \qw \\
  \lstick{\ket{0}} & \ghost{U_1} &  \ghost{U_2} & \ghost{U_3} & \qw & \lstick{\bra{0}}\\
}

We first introduce $HH=I$'s at the beginning and end of each line:

\hspace{5em} 
\Qcircuit @C=1em @R=.7em {
  \lstick{\ket{0}} & \gate{H} & \gate{H} & \multigate{2}{U_1} & \multigate{2}{U_2}  & \multigate{2}{U_3} & \gate{H} & \gate{H} & \meter     \qw \\
  \lstick{\ket{0}} & \gate{H} & \gate{H} & \ghost{U_1} & \ghost{U_2} & \ghost{U_3} & \gate{H} & \gate{H} & \meter    \qw \\
  \lstick{\ket{0}} & \gate{H} & \gate{H} & \ghost{U_1} &  \ghost{U_2} & \ghost{U_3} & \gate{H} & \gate{H} & \qw &\lstick{\bra{0}} \\
}

For intermediate $H$ gates (ones that are not on both ends), we apply the ``Hadamard gadget''. For example, we first transform on the second column of $H$ gates, teleporting the states to three new lines:

\hspace{5em} 
\Qcircuit @C=1em @R=.7em {
  \lstick{\ket{0}} & \gate{H} & \ctrl{3}  & \qw & \qw & \qw &  \qw & \qw & \qw & \gate{H} & \qw & \lstick{\bra{0}}\\
  \lstick{\ket{0}} & \gate{H} & \qw & \ctrl{3} & \qw & \qw & \qw & \qw & \qw & \gate{H} & \qw & \lstick{\bra{0}}\\
  \lstick{\ket{0}} & \gate{H} & \qw & \qw & \ctrl{3} & \qw & \qw  & \qw & \qw & \gate{H} & \qw &\lstick{\bra{0}}\\
  \lstick{\ket{0}} & \gate{H} & \gate{Z} & \qw & \qw & \multigate{2}{U_1} & \multigate{2}{U_2}  & \multigate{2}{U_3} & \gate{H} & \gate{H} & \meter\\
  \lstick{\ket{0}} & \gate{H} & \qw & \gate{Z} & \qw & \ghost{U_1} & \ghost{U_2}  & \ghost{U_3} & \gate{H} & \gate{H} & \meter\\
  \lstick{\ket{0}} & \gate{H} & \qw & \qw & \gate{Z} & \ghost{U_1} & \ghost{U_2}  & \ghost{U_3} & \gate{H} & \gate{H} & \qw & \lstick{\bra{0}} \\
}

and then transforming the the second to last column of $H$ gates:

\hspace{5em} 
\Qcircuit @C=1em @R=.7em {
  \lstick{\ket{0}} & \gate{H} & \ctrl{3}  & \qw & \qw & \qw &  \qw & \qw & \qw & \qw &  \qw & \gate{H} & \qw & \lstick{\bra{0}}\\
  \lstick{\ket{0}} & \gate{H} & \qw & \ctrl{3} & \qw & \qw & \qw & \qw & \qw & \qw &  \qw & \gate{H} & \qw & \lstick{\bra{0}}\\
  \lstick{\ket{0}} & \gate{H} & \qw & \qw & \ctrl{3} & \qw & \qw  & \qw & \qw & \qw  & \qw & \gate{H} & \qw &\lstick{\bra{0}}\\
  \lstick{\ket{0}} & \gate{H} & \gate{Z} & \qw & \qw & \multigate{2}{U_1} & \multigate{2}{U_2}  & \multigate{2}{U_3} & \ctrl{3} & \qw & \qw & \gate{H} & \qw & \lstick{\bra{0}}\\
  \lstick{\ket{0}} & \gate{H} & \qw & \gate{Z} & \qw & \ghost{U_1} & \ghost{U_2}  & \ghost{U_3} & \qw & \ctrl{3} & \qw & \gate{H} & \qw & \lstick{\bra{0}}\\
  \lstick{\ket{0}} & \gate{H} & \qw & \qw & \gate{Z} & \ghost{U_1} & \ghost{U_2}  & \ghost{U_3} & \qw & \qw & \ctrl{3} & \gate{H} & \qw & \lstick{\bra{0}}\\
  \lstick{\ket{0}} & \gate{H} & \qw & \qw & \qw & \qw & \qw  & \qw & \gate{Z} & \qw & \qw  &\gate{H} & \meter\\
  \lstick{\ket{0}} & \gate{H} & \qw & \qw & \qw & \qw & \qw & \qw  & \qw & \gate{Z} & \qw & \gate{H} & \meter\\
  \lstick{\ket{0}} & \gate{H} & \qw & \qw & \qw & \qw & \qw  & \qw &\qw & \qw & \gate{Z} & \gate{H} & \qw & \lstick{\bra{0}}\\
}\\

Thus for any $H$ gate inside $U_1$, $U_2$, and $U_3$, we can repeat the same steps and eliminate it by adding a new ancillary line. Because the transformed circuit has $H$ gates on each line at both ends, and only contains gates diagonal in the $Z$ basis in the middle, and have post-selected lines (one for each reduced intermediate $H$ gate) in the end, any post-selected BQP circuit can be reduced to a post-selected IQP circuit. Therefore $\text { post-BQP } \subseteq \text {post-IQP }$. 
\end{proof} 


\begin{theorem}[rephrase of \cite{Bremner_PH}]\label{thm:theorem-2}
\textit{If the output probability distributions generated by 
IQP circuits could be weakly classically simulated to within multiplicative error $1 \leq c<\sqrt{2}$ then post-IQP $\subseteq$ post-BPP.} 
\end{theorem}


The original proof is based on the analysis of bounded error without too much detail. \textbf{We expand the derivations with more clarity}:

\begin{proof} 
Let $L \in$ post-IQP be any language decided with bounded error by post-selecting some IQP circuit $\mathcal C_{w}$ with (single line) output registers $\mathcal{O}_{w}$ and post-selection registers $\mathcal{P}_{w}$. Then
$$
S_{w}(x)=\frac{\operatorname{Prob}\left[\mathcal{O}_{w}=x \,\&\, \mathcal{P}_{w}=0 \ldots 0\right]}{\operatorname{Prob}\left[\mathcal{P}_{w}=0 \ldots 0\right]},
$$
is the conditional probability of observing $x$. The authors define a variant of the bounded error condition that decides if a language is in post-IQP (compared with \ref{def: l_post_IQP}) (for some $0<\delta<\frac{1}{2}$):
\begin{equation}
\label{eq: bounded_error_1/2}
    \begin{split}
        &\text{if } w \in L  \text{ then } S_{w}(1) \geq \frac{1}{2}+\delta\\
        &\text{if } w \notin L \text{ then } S_{w}(1) \leq \frac{1}{2}-\delta.
    \end{split}
\end{equation}
Then, we apply the definition of multiplicative error: let $\mathcal{Y}_w$ be any subset of registers/lines of $\mathcal{C}_w$, if there is another classical randomized circuit $\tilde{\mathcal{C}}_w$ (thus of BPP) that can produce its output distribution to within a multiplicative error $c\geq 1$ relative to $\mathcal{C}_w$:
\begin{equation}
\label{eq: multiplicative_error}
    \frac{1}{c} \operatorname{Prob}\left[\mathcal{Y}_{w}=y_{1} \ldots y_{m}\right] \leq \operatorname{Prob}\left[\tilde{\mathcal{Y}}_{w}=y_{1} \ldots y_{m}\right] \leq c \operatorname{Prob}\left[\mathcal{Y}_{w}=y_{1} \ldots y_{m}\right].
\end{equation}

Additionally, define $\tilde{\mathcal{O}}_{w}$ and $\tilde{\mathcal{P}}_{w}$ as the output and post-selection registers of $\tilde{\mathcal{C}}_{w}$, we have a similar expression for:
\begin{equation}
    \tilde{S}_{w}(x)=\frac{\operatorname{Prob}[\tilde{\mathcal{O}}_{w}=x \, \& \, \tilde{\mathcal{P}}_{w}=0 \ldots 0]}{\operatorname{Prob}[\tilde{\mathcal{P}}_{w}=0 \ldots 0]}.
\end{equation}

According to Eq \ref{eq: multiplicative_error}, we have the following inequalities:
\begin{equation}
\begin{split}
    \frac{1}{c}\operatorname{Prob}\left[\mathcal{O}_{w}=x \, \&\,  \mathcal{P}_{w}=0 \ldots 0\right]\leq &\operatorname{Prob}[\tilde{\mathcal{O}}_{w}=x \,\&\, \tilde{\mathcal{P}}_{w}=0 \ldots 0] \leq c \operatorname{Prob}\left[\mathcal{O}_{w}=x \,\&\, \mathcal{P}_{w}=0 \ldots 0\right]\\
    \frac{1}{c} \operatorname{Prob}\left[\mathcal{P}_{w}=0 \ldots 0\right]\leq &\operatorname{Prob}[\tilde{\mathcal{P}}_{w}=0 \ldots 0] \leq c \operatorname{Prob}\left[\mathcal{P}_{w}=0 \ldots 0\right].
\end{split}
\end{equation}
Now we want to develop a similar bound for $\tilde{S}_{w}(x)=\frac{\operatorname{Prob}[\tilde{\mathcal{O}}_{w}=x \,\&\, \tilde{\mathcal{P}}_{w}=0 \ldots 0]}{\operatorname{Prob}[\tilde{\mathcal{P}}_{w}=0 \ldots 0]}$. Notice that, if we increase the numerator and decrease the denominator, the ratio increases:
\begin{equation}
    \tilde{S}_{w}(x)=\frac{\operatorname{Prob}[\tilde{\mathcal{O}}_{w}=x \,\&\, \tilde{\mathcal{P}}_{w}=0 \ldots 0]}{\operatorname{Prob}[\tilde{\mathcal{P}}_{w}=0 \ldots 0]}\leq \frac{c \operatorname{Prob}\left[\mathcal{O}_{w}=x \,\&\, \mathcal{P}_{w}=0 \ldots 0\right]}{\frac{1}{c} \operatorname{Prob}\left[\mathcal{P}_{w}=0 \ldots 0\right]}=c^2 S_{w}(x),
\end{equation}
and vice versa:
\begin{equation}
    \tilde{S}_{w}(x)=\frac{\operatorname{Prob}[\tilde{\mathcal{O}}_{w}=x \,\&\, \tilde{\mathcal{P}}_{w}=0 \ldots 0]}{\operatorname{Prob}[\tilde{\mathcal{P}}_{w}=0 \ldots 0]} \geq \frac{\frac{1}{c}\operatorname{Prob}\left[\mathcal{O}_{w}=x \,\&\, \mathcal{P}_{w}=0 \ldots 0\right]}{c \operatorname{Prob}\left[\mathcal{P}_{w}=0 \ldots 0\right]} = \frac{1}{c^2} S_{w}(x).
\end{equation}
The above two findings are organized to:
\begin{equation}
\label{eq: organized_inequality}
     \frac{1}{c^2} S_{w}(x) \leq \tilde{S}_{w}(x) \leq c^2 S_{w}(x).
\end{equation}
What are the conditions for $\tilde{S}_{w}(x)$ to also be able to detect if a bit string $w$ is from the language $L$ with bounded error (say, for $0 < \epsilon < \frac{1}{2}$) defined in \ref{eq: bounded_error_1/2}? We look at two cases:
\begin{itemize}
    \item If $w\in L$, ${S}_{w}(1)\geq \frac{1}{2}+\delta$. Then plugging in \ref{eq: organized_inequality}, $\frac{\frac{1}{2}+\delta}{c^2}\leq \tilde{S}_{w}(1)$. For $\tilde{S}_{w}(x)$ to have bounded error $\epsilon$ (refer to \ref{eq: bounded_error_1/2}), in this subcase we require $\tilde{S}_{w}(1)\geq \frac{1}{2}+\epsilon$, which is guaranteed only when $\frac{\frac{1}{2}+\delta}{c^2} \geq \frac{1}{2}+\epsilon\implies \frac{\frac{1}{2}+\delta}{c^2} > \frac{1}{2}$. 
    
    \item If $w\notin L$, ${S}_{w}(1)\leq \frac{1}{2}-\delta$. Then plugging in \ref{eq: organized_inequality}, $\tilde{S}_{w}(1)\leq c^2 (\frac{1}{2}-\delta)$. For $\tilde{S}_{w}(x)$ to have bounded error $\epsilon$ (refer to \ref{eq: bounded_error_1/2}), in this subcase we require $\tilde{S}_{w}(1)\leq \frac{1}{2}-\epsilon$, which is guaranteed only when $c^2 (\frac{1}{2}-\delta) \leq \frac{1}{2}-\epsilon\implies c^2 (\frac{1}{2}-\delta) < \frac{1}{2}$. 
\end{itemize}
Therefore, organize the two cases we have:
\begin{equation}
    \begin{split}
        \frac{\frac{1}{2}+\delta}{c^2} > \frac{1}{2}, \\
        c^2 (\frac{1}{2}-\delta) < \frac{1}{2},
    \end{split}
\end{equation}
which is reduced to 
\begin{equation}
    1+2\delta > c^2.
\end{equation}
Because $0<\delta<\frac{1}{2}$ and $c>1$, we have $c^2<2\implies c< \sqrt{2}$ will guarantee that $L \in \text {post-BPP}$. 
\end{proof} 

The following corollary is based on the previous two theorems:

\begin{corollary}[rephrase of \cite{Bremner_PH}]
\label{cor:corollary_1}
If the output probability distributions generated by 
IQP circuits could be weakly classically simulated to within multiplicative error $1 \leq c<\sqrt{2}$ then post-BPP = PP.
\end{corollary}

\begin{proof} 
Because we have shown post-IQP = post-BQP = PP, if we can classically simulate an IQP circuit to within multiplicative error $1 \leq c<\sqrt{2}$ then post-IQP = post-BQP $\subseteq$ post-BPP (by Theorem~\ref{thm:theorem-2}). In addition, with post-BPP $\subseteq$ post-BQP \cite{eecs_BPP}, (if the multiplicative error simulation condition holds) we arrive at the relation post-BPP = post-BQP = PP. 
\end{proof} 

\begin{corollary}[rephrase of \cite{Bremner_PH}]
The polynomial hierarchy would collapse to its third level, i.e. $\mathrm{PH}=\Delta_{3}$, if the stated assumption in Corollary~\ref{cor:corollary_1} holds.
\end{corollary}
\begin{proof}
Using the fact that $$\mathrm{P}^{\text {post }-\mathrm{BPP}} \subseteq \mathrm{P}^{\Delta_{3}}=\Delta_{3}$$ and Toda's Theorem \cite{boaz_complexity}, $$\mathrm{PH} \subseteq \mathrm{P}^{\mathrm{PP}},$$
if the efficient classical simulation of IQP circuit assumption holds (Corollary~\ref{cor:corollary_1}), we get post-BPP = PP, which means that 
\begin{equation}
    \mathrm{PH} \subseteq \mathrm{P}^{\mathrm{PP}}=\mathrm{P}^{\text {post }-\mathrm{BPP}} \subseteq \mathrm{P}^{\Delta_{3}}=\Delta_{3},
\end{equation}
implying $\mathrm{PH} = \Delta_3$.
\end{proof} 

We thus conclude the first part of the proof review by a quick summarization: if the IQPC can be efficiently simulated classically (to within a multiplicative error), then the polynomial hiearchy will collapse to the third level, which is unlikely!

\subsection{Part 2: MPQ circuits generalize IQP circuits}
This part is based on \cite{expressive_q}, which shows that MPQCs are more expressive than IQPCs.  
Intuitively, the proof in \cite{expressive_q} uses substitution tricks to transform elements in the IQP circuit (i.e., the universal gate set $\{H, T, Z, CZ\}$) to those of MPQC (single-bit rotations, phase rotation, and CNOTs). We organize the proof in the similar ``divide and conquer'' approach, substituting each component of IQP circuit by those of MPQCs. We stress again that in a MPQC block, the single-bit unitary $U\left(\theta_{l}^{j}\right)=R_{z}\left(\theta_{l}^{j, 1}\right) R_{y}\left(\theta_{l}^{j, 2}\right)R_{z}\left(\theta_{l}^{j, 3}\right) R_{\phi}\left(\phi \right) R_{x}\left(\theta_{l}^{j, 4}\right) R_{z}\left(\theta_{l}^{j, 5}\right) R_{x}\left(\theta_{l}^{j, 6}\right)$ is parameterized by the rotation and phase angles $\{\theta_l^{j,1}, \theta_l^{j,2}, \theta_l^{j,3}, \phi_l^{j}, \theta_l^{j,4}, \theta_l^{j,5}, \theta_l^{j,6}\}$. The CNOT gates follow the following connection pattern: $CNOT_{12}, CNOT_{13},\dots CNOT_{1N}$ with $N-1$ entanglement connections where $N$ is the number of qubits.

\begin{proposition}
\label{prop: H}
The $H$ gates at the two ends (first and final layers) of an IQP circuit can be expressed by two MPQC blocks.
\end{proposition}

\begin{proof}
Because $H=R_{x}(\pi / 2) R_{z}(\pi / 2) R_{x}(\pi / 2)$, we can simply set the MPQC angle parameters to $\{\theta_l^{j,1}, \theta_l^{j,2}, \theta_l^{j,3}, \phi_l^{j}, \theta_l^{j,4}, \theta_l^{j,5}, \theta_l^{j,6}\}=(0,0,0,0, \pi / 2, \pi / 2, \pi / 2)$ and then set the following block to have parameters to $\{\theta_l^{j,1}, \theta_l^{j,2}, \theta_l^{j,3}, \phi_l^{j}, \theta_l^{j,4}, \theta_l^{j,5}, \theta_l^{j,6}\}=(0,0,0,0, 0, 0, 0)$ so that the CNOT entanglement portions of the these two blocks cancel out to $I$ (Note that $(CNOT_{12}, CNOT_{13},\dots CNOT_{1N}) (CNOT_{12}, CNOT_{13},\dots CNOT_{1N})=I$, because $CNOTs$ commute if only the control bit is the same).
\end{proof} 

\begin{proposition}
\label{prop: T}
The $T$ gates of an IQP circuit can be expressed by two MPQC blocks.
\end{proposition}
\begin{proof}
Because $T=R_{\phi}(\pi/4)$, then we can set a MPQC block's parameters to $\{\theta_l^{j,1}, \theta_l^{j,2}, \theta_l^{j,3}, \phi_l^{j}, \theta_l^{j,4}, \theta_l^{j,5}, \theta_l^{j,6}\}=(0,0,0,\pi/4, 0,0,0)$ followed by another block with $\{\theta_l^{j,1}, \theta_l^{j,2}, \theta_l^{j,3}, \phi_l^{j}, \theta_l^{j,4}, \theta_l^{j,5}, \theta_l^{j,6}\}=(0,0,0,0, 0,0,0)$ to cancel out the CNOTs.
\end{proof} 

\begin{proposition}
\label{prop: CZ}
The $CZ$ gates of an IQP circuits can be expressed by constant numbers of MPQC blocks.
\end{proposition} 
\begin{proof} 
The first challenge is that the arbitrary $CZ_{k,j}$ connectivity (qubit $k$ controls qubit $j$) might not be covered in the pre-defined connectivity $(CNOT_{12}, CNOT_{13},\dots CNOT_{1N})$, so we use the following steps:
\begin{itemize}
    \item We first use the $\textit{SWAP}_{1,k}$ operation (expressed by MPQC blocks) to switch the control bit ($k$-th qubit) of the arbitrary $CZ_{k,j}$ to the first bit. A single $\textit{SWAP}_{1,k}$ gate can be expressed by three \textit{CNOTs}: $SWAP_{1, k}=CNOT_{1, k}CNOT_{k, 1}CNOT_{1, k}$. Following Proposition 1 in \cite{expressive_q}, the $CNOT_{1,k}$ gate can be substituted by $C_k (CNOT_{1,k}) B_k (CNOT_{1,k}) A_k$, where $A_k=R_z (0)R_y(\pi/2)$, $B_k=R_y(-\pi/2)R_z(\pi)$, and $C_k=R_z(\pi)$. Therefore, we first use three MPQC blocks to parameterize $A_k$, $B_k$, and $C_k$ on the $k$-th qubit only, setting all angle parameters to $0$ in other qubits. Then another MPQC with all $0$ parameters is appended after the three blocks to \textbf{cancel out} the $CNOT$s that are irrelevant to the $k$-th qubit. The resulting circuit, is exactly equivalent to $C_k (CNOT_{1,k}) B_k (CNOT_{1,k}) A_k=CNOT_{1,k}$.
    
    \item Then, the reverse $CNOT_{k, 1}$ gate in the middle of the $\textit{SWAP}_{1,k}$ is substituted by $CNOT_{k, 1} = (H\otimes H) CNOT_{1, k} (H\otimes H)$, where each $(H\otimes H)$ gate can be substituted by two MPQC blocks according to Proposition \ref{prop: H}.
    \item Then, because the $CZ_{1,j}$ gate (after transfering the control from qubit $k$ to qubit $1$) is equivalent to $CZ_{1,j}=(I \otimes H) CNOT_{1,j} (I\otimes H)$, we already show how to implement $H$ and $CNOT_{1,j}$ gates, according to the previous steps and Proposition \ref{prop: H}.
\end{itemize}
Therefore, we the whole process is broken down to
\begin{equation}
    \begin{split}
        &CZ_{k,j}=SWAP_{1,k}CZ{1,j} SWAP_{1,k}\\
        &=\underbrace{CNOT_{1, k}CNOT_{k, 1}CNOT_{1, k}}_{SWAP_{1,k}} \underbrace{(I \otimes H) CNOT_{1,j} (I\otimes H)}_{CZ_{1,j}} \underbrace{CNOT_{1, k}CNOT_{k, 1}CNOT_{1, k}}_{SWAP_{1,k}}\\
        &=\underbrace{CNOT_{1, k}(H\otimes H) CNOT_{1, k} (H\otimes H)CNOT_{1, k}}_{SWAP_{1,k}} \\
        &\underbrace{(I \otimes H) CNOT_{1,j} (I\otimes H)}_{CZ_{1,j}} \underbrace{CNOT_{1, k}(H\otimes H) CNOT_{1, k} (H\otimes H) CNOT_{1, k}}_{SWAP_{1,k}},
    \end{split}
\end{equation}
where in the last step, all the components can be substituted by constant numbers of MPQC blocks ($2$ blocks for each $I\otimes H$ or $H\otimes H$ and 4 blocks for each $CNOT_{1,i}$).
\end{proof} 

Therefore, for circuits on $N$ qubits, Propositions \ref{prop: H}, \ref{prop: T}, and \ref{prop: CZ} together show that we need at most $\mathcal{O}(poly(N))$ MPQC blocks to completely substitute an IQP circuit with $\mathcal{O}(poly(N))$ commuting gates, as each commuting gate (as well as the $H$ gate) can be substituted with constant numbers of MPQC blocks. Therefore, the expressive power of MPQC is at least that of IQPC, which is more expressive than classical circuits (such as NNs).

\section{ML Methods to Showcase the Power of MPQCs}\label{ML_methods}
With the theoretical knowledge that MPQCs are more expressive than classical circuits, we wish to actually build an MPQC-based generative model and compare its performance with a classical NN on the task of learning some high-dimensional discrete dataset. We first introduce relevant methods used in our models and experiments.

\subsection{Gumbel Softmax for classical NNs}
The Gumbel softmax \cite{gumbel_softmax} is a popular method for a generative model to output discrete/categorical data. It approximates the non-differentiable sample from a categorical distribution with a differentiable sample from the Gumbel-Softmax distribution, which can be smoothly annealed into a categorical distribution during training with a temperature schedule. Previous works have shown its effectiveness in GAN \cite{Gumbel_GAN}, variational autoencodrs and generative semi-supervised classification \cite{gumbel_softmax}, self-supervised discrete speech representations \cite{gumbel_speech}, and community clustering detection \cite{gumbel_clustering}, etc. Specifically, Gumbel Softmax \cite{gumbel_softmax} uses the softmax function as a continuous, differentiable approximation to arg max \cite{Gumbel_original}, and generates $k$-dimensional sample vectors $y$, where
$$
y_{i}=\frac{\exp \left(\left(\log \left(\pi_{i}\right)+g_{i}\right) / \tau\right)}{\sum_{j=1}^{k} \exp \left(\left(\log \left(\pi_{j}\right)+g_{j}\right) / \tau\right)} \quad \text { for } i=1, \ldots, k.
$$

One can thus use the Gumbel Softmax in a generative model to output discrete synthesized data and back-propagate through the softmax operation.

\subsection{Quantum Born Machine}
Quantum circuit Born machine (QCBM) \cite{benedetti-quantum-shallow}\cite{dif_born} leverages the probabilistic nature of quantum systems for generating classical data (collapsed state upon measurement). QCBM represents classical probability distribution via the Born's rule \cite{born_boltzmann}\cite{born_rule}. QCBM utilizes a variational quantum circuit (such as MPQC) to evolve the initial/input quantum state $|z\rangle=|0\rangle^{\otimes N}$ to some target state via unitary gates: $\left|\psi_{\theta}\right\rangle=U_{\theta}|z\rangle$. One measures the outputs state in the computational basis to produce a classical sample (bit string) $x \sim p_{\theta}(x)=\left|\left\langle x \mid \psi_{\theta}\right\rangle\right|^{2}$. Excitingly, the output probability densities of a general quantum circuit cannot be efficiently simulated by classical means, the QCBM is among the several proposals to show quantum supremacy \cite{boixo_supremacy}. Ref. \cite{dif_born} developed a differentiable learning scheme of QCBM by minimizing the maximum mean discrepancy (MMD) loss  \cite{dif_born} \cite{MMD_gretton} using an Gaussian Kernel:

\begin{equation}
    \begin{aligned}
\mathcal{L} &=\left\|\sum_{x} p_{\theta}(x) \phi(x)-\sum_{x} \pi(x) \phi(x)\right\|^{2} \\
&=\underset{x \sim p_{\theta}, y \sim p_{\theta}}{\mathbb{E}}[K(x, y)]-2 \underset{x \sim p_{\theta}, y \sim \pi}{\mathbb{E}}[K(x, y)]+\underset{x \sim \pi, y \sim \pi}{\mathbb{E}}[K(x, y)],
\end{aligned}
\end{equation}

where $\pi$ and $p_\theta$ are the data and model distributions respectively. The function $\phi(\cdot)$ maps a data sample to an infinite-dimensional reproducing kernel Herbert space \cite{kernel_space}. This setup conveniently saves one from working in the infinite-dimensional feature space by defining a kernel function $K(x, y)=\frac{1}{c} \sum_{i=1}^{c} \exp \left(-\frac{1}{2 \sigma_{i}}|x-y|^{2}\right)$ to only reveal the pair-wise distances between two samples under various scales (the \textit{kernel trick}).

\subsection{Generative Adversarial Quantum Circuits}
As a classical-quantum-hybrid architecture, Generative Adversarial Quantum Circuits \cite{adv_Born} uses the QCBM as its generator and a classical deep neural network as the Discriminator. Similar to that of a classical GAN, it follows a adversarial training scheme: the (classical) discriminator $\mathrm{D}$ takes either real samples from the dataset or synthetic samples from the quantum circuit and outputs a ``score'',  $D_{\phi}(x)$ that predicts how likely the sample is from the data distribution. The Generator (G) generates synthetic samples to ``fool'' the discriminator by producing samples $x$ that have high scores $D_{\phi}(x)$. The quantum generator and the classical discriminator play against each other according some loss function, which defines a minimax problem:
$$
\min _{G_{\theta}} \max _{D_{\phi}} \mathbb{E}_{x \sim \pi}\left[\ln D_{\phi}(x)\right]+\mathbb{E}_{x \sim p_{\theta}}\left[\ln \left(1-D_{\phi}(x)\right)\right].
$$

This zero-sum minimax objective \cite{goodfellow2014generative} does not perform well for classical GANs due to vanishing gradients. The non-saturating loss function is:
\begin{equation}
    \begin{aligned}
\mathcal{L}_{D_{\phi}} &=-\mathbb{E}_{x \sim \pi}\left[\ln D_{\phi}(x)\right]-\mathbb{E}_{x \sim p_{\theta}}\left[\ln \left(1-D_{\phi}(x)\right)\right] \\
\mathcal{L}_{G_{\theta}} &=-\mathbb{E}_{x \sim p_{\theta}}\left[\ln D_{\phi}(x)\right].
\end{aligned}
\end{equation}

\subsection{Maximal Coding Rate Objective} 
We propose to use the Maximal Coding Rate Reduction (MCR) \cite{yu2020learning} objective to train the quantum generative model. The MCR objective has interesting geometric interpretation related to ball counting in information theory \cite{chan2020redu}\cite{yu2020learning}. When the loss is maximized, data samples from two distinct classes are maximally compressed to to orthogonal subspaces, as opposed to falling into two sides of a separable hyperplane. This geometric feature is intuitively a more natural way of compressing data. We consider this distance measure exactly because of this richer feature space representation than the original GAN (which only outputs a $1$-dimensional score). In the simplest case (discrete and compact distributions $ p_X$ and $ p_Y$), given two groups of samples, $x\sim p_X$ and $y\sim p_Y$, and if we treat each sample as a column vector and construct $\X$ and $\Y$, both are $\mathcal{R}^{D\times m}$, where $D$ is the dimension of features and $m$ is the batch size, we define the MCR (second order) ``distance'' between sample matrices $\X$ and $\Y$ as:
\begin{equation}
\label{MCR_x_space}
    \begin{split}
        \Delta R([\X,\Y])=&\frac{1}{2}\log\det(\I + \frac{d}{2m\epsilon^2}\X\X^T+\frac{d}{2m\epsilon^2}\Y\Y^T)\\
        &-\frac{1}{4}\log\det(\I + \frac{d}{m\epsilon^2}\X\X^T)-\frac{1}{4}\log\det(\I + \frac{d}{m\epsilon^2}\Y\Y^T).
    \end{split}
\end{equation}
This metric is $0$ iff $\frac{1}{m}\X\X^T=\frac{1}{m}\Y\Y^T$, i.e., the sample covariance matrices (given the same batch sizes) are equal. Since the binary images/bitstrings follow distributions with higher moment correlations than the second moment, it is not meaningful to directly apply Eq. \ref{MCR_x_space} to the data space. Rather, we use a parametric mapping $D_\phi$ (which is essentially the Discriminator) to compress the high-dimensional discrete data to a low-dimensional continuous space $D_\phi: X\rightarrow Z$, where $x\in \mathbb{R}^D$ and $x\in \mathbb{R}^d$, $D\gg d$. In addition, if we express the sample covariance matrix in terms of probability $p(x)$ (for the case when the sample size $m\rightarrow \infty$):
\begin{equation}
    \lim_{m\rightarrow \infty} \frac{d}{m\epsilon^2}\X\X^T=\frac{d}{\epsilon^2}\sum_{\x} p_{\X}(\x)\x\x^T,
\end{equation}
the loss function becomes:
\begin{equation}
\label{infinite}
    \begin{split}
        \Delta& \lim_{m\rightarrow \infty} R([\X,\Y])=\frac{1}{2}\log\det(\I + \frac{d}{2\epsilon^2}\sum_{\x} p_{\X}(\x)D_\phi(\x)D_\phi(\x)^T+\frac{d}{2\epsilon^2}\sum_{\y} p_{\Y}(\y)D_\phi(\y)D_\phi(\y)^T)\\
        &-\frac{1}{4}\log\det(\I + \frac{d}{\epsilon^2}\sum_{\x} p_{\X}(\x)D_\phi(\x)D_\phi(\x)^T)-\frac{1}{4}\log\det(\I + \frac{d}{\epsilon^2}\sum_{\y} p_{\Y}(\y)D_\phi(\y)D_\phi(\y)^T).
    \end{split}
\end{equation}

\subsection{Differential Circuit Gradients}
The MMD loss, non-saturating GAN loss, and the MCR loss all have differentiable gradients with respect to the circuit gate parameters (rotation angles). According Ref. \cite{quantum_gradient_36}, \cite{quantum_gradient_38}, the gradient of the probability on value $p_\theta (x)$ with respect to a quantum circuit parameter $\theta$, is
\begin{equation}
\label{gradient}
    \frac{\partial p_{\boldsymbol{\theta}}(x)}{\partial \theta_{l}^{\alpha}}=\frac{1}{2}\left(p_{\boldsymbol{\theta}^{+}}(x)-p_{\boldsymbol{\theta}^{-}}(x)\right)
\end{equation}

where $p_{\theta^\pm} (x)=\left|\left\langle x \mid \psi_{\theta}\right\rangle\right|^{2}$ is the probability of observing $x$ generated by the same circuit with shifted circuit parameter $\theta^{\pm} \equiv \theta \pm \frac{\pi}{2}$. This is an unbiased estimate of the gradient.

\subsection{Sampling-based loss function gradients} 
All three loss function that we consider in this study, namely the MMD, non-saturating GAN and the MCR losses are sampling-based. To see this, we apply the gradient formula \ref{gradient} to all three objective functions with respect to a given generator parameter $\theta$ and get (for derivations, see Appendix B and C for the NN and RBF kernel versions): 
\begin{equation}
    \begin{aligned}
\nabla_{\theta}\mathcal{L}_{\text{MMD}} &=\underset{x \sim p_{\theta^{+}}, y \sim p_{\theta}}{\mathbb{E}}[K(x, y)]-\underset{x \sim p_{\theta^{-}}, y \sim p_{\theta}}{\mathbb{E}}[K(x, y)]-\underset{x \sim p_{\theta^{+}}, y \sim \pi}{\mathbb{E}}[K(x, y)]+\underset{x \sim p_{\theta^{-}}, y \sim \pi}{\mathbb{E}}[K(x, y)].
\end{aligned}
\end{equation}
And for the non-saturating GAN loss:
\begin{equation}
\begin{split}
    \nabla_{\theta} \mathcal{L}_{G_{\theta}}=\frac{1}{2} \mathbb{E}_{x \sim p_{\theta^{-}}}\left[\ln D_{\phi}(x)\right]-\frac{1}{2} \mathbb{E}_{x \sim p_{\theta^{+}}}\left[\ln D_{\phi}(x)\right],
\end{split}
\end{equation}
then for the MCR loss:
\begin{equation}
    \nabla_{\theta}\mathcal{L}_{\text{MCR}}=\nabla_{\theta}\mathcal{L}_{\text{A}}+\nabla_{\theta}\mathcal{L}_{\text{B}},
\end{equation}
\begin{equation}
    \begin{split}
        \nabla_{\theta}\mathcal{L}_{\text{A}}
        &=\frac{1}{4}\mathbb{E}_{\hat{\x}\sim p_{\theta^+}}<(\I + \frac{d}{2\epsilon^2}\sum_{\x} p_{\boldsymbol{\theta}}(\x)\phi(\x)\phi(\x)^T+\frac{d}{2\epsilon^2}\sum_{\y} p_{\Y}(\y)\phi(\y)\phi(\y)^T)^{-1}, \phi(\hat{\x})\phi(\hat{\x})^\top>\\
        &-\frac{1}{4}\mathbb{E}_{\hat{\x}\sim p_{\theta^-}}<(\I + \frac{d}{2\epsilon^2}\sum_{\x} p_{\boldsymbol{\theta}}(\x)\phi(\x)\phi(\x)^T+\frac{d}{2\epsilon^2}\sum_{\y} p_{\Y}(\y)\phi(\y)\phi(\y)^T)^{-1}, \phi(\hat{\x})\phi(\hat{\x})^\top>,
    \end{split}
\end{equation}
\begin{equation}
    \begin{split}
        \nabla_{\theta}\mathcal{L}_{\text{B}}&=\frac{1}{8}\mathbb{E}_{\hat{\x}\sim p_{\theta^+}}<(\I + \frac{d}{\epsilon^2}\sum_{\x} p_{\boldsymbol{\theta}}(\x)\phi(\x)\phi(\x)^T)^{-1}, \phi(\hat{\x})\phi(\hat{\x})^\top>\\
        &-\frac{1}{8}\mathbb{E}_{\hat{\x}\sim p_{\theta^-}}<(\I + \frac{d}{\epsilon^2}\sum_{\x} p_{\boldsymbol{\theta}}(\x)\phi(\x)\phi(\x)^T)^{-1}, \phi(\hat{\x})\phi(\hat{\x})^\top>,
    \end{split}
\end{equation}
where $<A,B>$ denotes the element-wise dot product between matrices $A$ and $B$.

A major mistake or misconception some authors have made in the past is the (wrong) choice of \textbf{intractable loss functions}. The amplitudes $\alpha_x$ in the quantum state, $\ket{\psi}_n = \sum_{x=0}^{2^n} \alpha_x \ket{x}$ are unknown to the observer, and can only be approximated by tallying many measurement results. Loss functions that explicitly require generated sample probability $p_G(x)$ for generated sample $x$, such as the KL divergence and negative log-likelihood, \textbf{cannot be efficiently approximated by sampling alone}, especially when the data is high-dimensional and multi-modal, as seen in natural images.
\\

\subsection{Arguments against MMD Loss (Maximal Mean Discrepancy)}
Although MMD(P, Q) is a perfect statistic (tautology) to see if probability distribution P equals Q, test power of its empirical estimation depends on the form of used kernels \cite{binkowski2018demystifying} \cite{liu2021learning}. Gaussian kernels have limited representation power in multi-dimensional samples. It requires many observations to distinguish the two distributions \cite{liu2021learning} with Gaussian Kernel. Arguments against MMD include:
\begin{enumerate}
    \item ``The Gaussian kernel (used by previous MMD-based adversarial data detection methods) has limited repre- sentation power and cannot measure the similarity between two multidimensional samples (e.g., images) well'' \cite{gao2021maximum} \cite{wenliang2021learning}
    \item ``Previous MMD-based adversarial data detection methods overlook the optimization of parameters of the used kernel.''\cite{gao2021maximum}.  Furthermore, recent studies have shown that \textbf{Gaussian kernel with an optimized bandwidth still has limited representation power} for complex distributions (e.g., multi- modal distributions used in (\cite{wenliang2021learning}, \cite{liu2021learning})). 
    
    \item ``Although such (spatially invariant, i.e., Gaussian) kernels are sufficient for consistency in the infinite-sample limit, the induced models can fail in practice on finite datasets, especially if the data takes differently-scaled shapes in different parts of the space.''\cite{wenliang2021learning}
\end{enumerate}

\subsection{Discriminator or Kernel?}
With ``universal'' RBF kernels such as the Gaussian or Polynomial Kernel, when the sample size is infinite, theoretically we have that two distributions are identical iff the MMD loss is zero \cite{MMD_gretton}. However, as shown in our experiments, in practice when we only have limited sample sizes, MMD's performance degrades drastically. Hence an extra classical neural network is needed to (implicitly) gauge the distance between the real and the generated data distributions with more robustness.

\section{Numerical Simulations} \label{numerical_sims}
We simulate the quantum circuit on a classical machine using a custom built quantum simulator based on Numpy and the Scipy sparse matrices package (scipy.sparse). For the adversarial architecture, we build the classical Discriminator with PyTorch and use Adam as our gradient-based optimizer.

\paragraph{Architectures} For fair comparisons, we use a $4$-layer neural network for the classical Generator with $56$ parameters. We choose a latent vector space of $2$ dimensions, and the Discriminator takes a similar, but reverse layer design as the Generator. 

Meanwhile, for the variational quantum circuit, we use a depth-$4$ architecture, where each depth/layer consists of one single-rotation gate on all qubits and a pair-wise entanglement layer with the CNOT gate. There are in total $52$ parameters.

\paragraph{Dataset} We train all of the quantum models with the same randomly initialized state by setting the Numpy seed to a value of $700$. All models are trained till convergence, and we present the results of both the best model and the last-iteration model. We evaluate the performance of the model by the Total Variation (TV) metric:
\begin{equation}
    T V\left(p_{\theta}, \pi\right)=\frac{1}{2} \sum_{\boldsymbol{x}}\left|p_{\theta}(\boldsymbol{x})-\pi(\boldsymbol{x})\right|,
\end{equation}
where $p_\theta$ is the probability density of the trained model, and $\pi$ is the data probability density. For the quantum Generator circuit, we have direct access to the $p_\theta(x)$ values from the quantum simulator, meanwhile in the case of the classical GAN+Gumbel Softmax model, we approximate $\p_\theta(x)$ by taking large enough amount of samples after training.

We use the $2\times2$ Bars and Stripes (BAS) dataset, which is a binary-pixel dataset consists of bars (vertical) and stripes (horizontal). For $n \times m$ pixels, the number of patterns that are valid BAS is $N_\text{BAS}= 2^n +2^m-2$. The target/data distribution, $p(x)=\frac{1}{N_\text{BAS}}$, if and only if $x$ is a valid BAS image. 
\begin{center}
\begin{figure}
    \centering
    \includegraphics[scale=0.3]{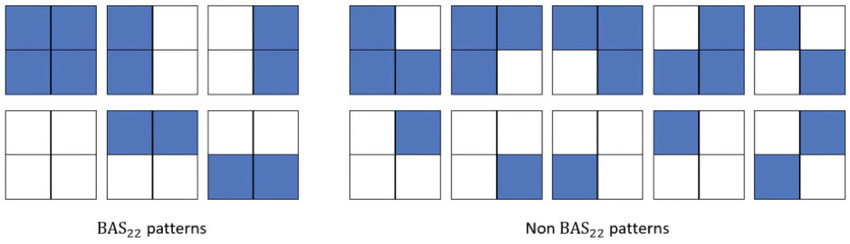}
    \caption{BAS dataset. Image from \cite{BAS}}
    \label{fig: BAS}
\end{figure}
\end{center}

\subsection{Experiment 1: Classical GAN+Gumbel Softmax}
We first investigate the discrete data learning capability of the classical GAN+Gumbel Softmax model. 
The randomly initialized Generator achieves a TV value of $1.304$. We first use constant temperatures of 1e-2 and 1e-4 for the softmax operation, and then experiment with a linear annealing schedule from 1e-2 to 1e-4. We train the model $2000$ iterations until convergence, and the learning outcome (after sampling $16000$ generated samples and approximate the distribution) achieves a TV value of $1.13, 1.62, $ and $1.29$ respectively (Fig \ref{fig: gumbel}). For constant temperatures, the Generator barely captures one mode of the real distribution, producing much noise on invalid patterns. The situation does not get better as we anneal the temperature schedule. Modal collapse and noises persist.

\subsection{Experiment 2: Varying sample sizes with Born Machine+Gaussian Kernel+MMD loss}
In this experiment, we investigate how the Born Machine+Gaussian Kernel+MMD loss scheme performs under the condition of limited batch size. This has never been investigated in literature, as previous authors simply use sample sizes that far exceeds the probability space itself. We train the same architecture starting from the same randomly initialized circuit whose probability distribution has a TV value of $1.064$. We use a learning rate of $0.001$ and train each model until convergence. The TV metric value with respect to the number of samples per batch is shown in Fig \ref{fig: kernel}.

The performance of this scheme drops drastically as the sample batch size decreases, as reflected by the drastic increase in total variation. As the sample size decreases, some modes are experiencing the modal collapse issue, and the circuit assigns noisy probabilities to invalid patterns. Thus this scheme faces serious practical concerns, especially when one deals with much higher-dimensional data and larger probability space, which is often intractable as in the case of natural images.

\subsection{Experiment 3: Varying sample sizes with Born Machine+Adversarial Training+Non-saturating GAN loss}
We now repeat the experiment with the ``Born Machine+Adversarial Training+Non-saturating GAN loss'' scheme. Although GAN's are known to be hard to train, to our surprise, this scheme is much more robust and stable than expected. We update the Discriminator (D) twice for each Generator (G) update, We use a learning rate of 1e-3 for both the Discriminator (D) and the Generator (G). All models are trained with a learning rate of 1e-3, except for the case with batch size = 4, which uses a learning rate of 1e-4. The TV metric values with respect to the number of samples per batch is shown in Fig \ref{fig: vanilla}.

Quite different from the previous scheme, there does not seem to be a clear correlation between the total variation performance and the sample sizes per batch, as a smaller batch size might even enhance the training result. Another important observation is that adversarial training helps eliminate the generation of invalid (non-BAS) image patterns, which is reflected by the low noise level of the learnt distribution (orange line) in all the plots. However, serious modal collapse is present in all batch sizes, which confirms with the observation in \cite{adv_Born}, where the authors trained on the $3 \times 3$ BAS dataset with a batch size of $512$ (equal to the total number of patterns produced by $9$ qubits).

\subsection{Experiment 4: Two-step fine-tuning scheme}
How can we improve the training result, i.e., how to decrease the noise in the learnt distribution, meanwhile avoiding modal collapse? We propose a two-step training and fine-tuning scheme, utilizing the discoveries in the previous two experiments. Namely, we first train the adversarial scheme (Born Machine+Adversarial Training+Non-saturating GAN loss) until convergence, which ensures a good starting point where invalid patterns are avoided. We then fine-tune the circuit with the Gaussian kernel (Born Machine+Gaussian Kernel+MMD loss) scheme which does not involve a neural network Discriminator. We herein showcase the effectiveness of this fine-tuning scheme by selecting the worst-learnt circuit in Experiment 3., which is the case with a batch size=64 with TV=0.502. We fine-tune the circuit with the Gaussian kernel, with a learning rate of 1e-4. After another $2000$ iterations, the model arrives at a configuration that produces a TV=0.089, which is the best performance we have ever witnessed (Fig \ref{fig: finetune}), even outperforming the Gaussian Kernel+MMD loss scheme with the infinite samples case.

\subsection{Experiment 5: Born Machine+Adversarial training+MCR loss}
We now experiment with our proposed MCR loss function in the adversarial training setting. We use different batch sizes and investigate the modal collapse as well as the noise issue encountered before. We discovered that the MCR loss function is not only more computationally heavy, but is also less stable than the non-saturating GAN loss. The learnt models are shown in Fig \ref{fig: mcr}. Although the MCR loss does not completely solve the modal collapse issue, the learning result is superior than that of the non-saturating GAN loss in terms of both TV and modal collapse (but it does not surpass the fine-tuning method). The Discriminator does map the real and synthesized data into almost orthogonal linear subspaces at the early stage, and then the quantum generator learns to assign higher probabilities to the valid patterns in order to align the synthesized subspace to the real subspace (Fig \ref{fig: subspace}). Additionally, a larger sample size does result in better learning in terms of the TV score.

\section{Conclusion and Future Work} \label{conclusion}
In this study, we first reviewed the proof that MPQCs can output probability distributions that cannot be efficiently simulated by classical means, unless the Polynomial Hierarchy collapses to the third level. We organized the proof in a detailed and intuitive manner.

Then, to experimentally showcase the power of MPQC over classical NNs, we systematically studied the performance of MPQC-based quantum generative models under limited sample sizes and experimentally showed that quantum circuits are indeed better than the GAN+Gumbel Softmax counterpart at learning discrete data. We confirmed that the RBF (Gaussian) Kernel+MMD loss scheme performs poorly when sample sizes are much smaller than the size of the probability space, which is unfortunately often the case for natural image learning. We discovered that adversarial training with the non-saturating GAN loss is stable and robust against limited sample sizes, and it helps eliminate the generation of invalid patterns, but meanwhile suffering from modal collapse. We then proposed a sample-efficient, two-step training-and-fine-tuning scheme that first uses the quantum-classical GAN to land on a good intermediate result (eliminating invalid patterns), and then fine-tunes the circuit with the Gaussian Kernel+MMD loss scheme. This two-step training scheme achieves the best ever result among all models considered in this study. Lastly, we experimented a novel information-theoretic MCR loss function and discovered that it does not completely solve the modal collapse issue, although it outperforms the non-saturating GAN loss.

Future work involves dissecting the failing modes of the MCR loss function (stucking at local minima, ill-conditioned scenarios, etc) and coming up with a more efficient training algorithm specific for MCR. 

\newpage
\appendix
\appendixpage
\section{Figures}
\begin{figure}
    \centering
    \includegraphics[scale=0.2]{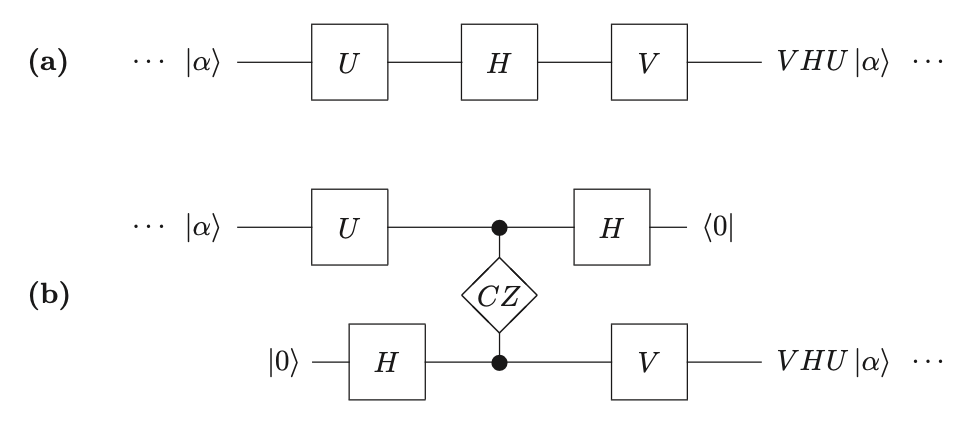}
    \caption{The Hadamard gadget for ``teleporting'' qubit A (top line) to qubit B (bottom line). Image from \cite{Bremner_PH}.}
    \label{fig:H_gadget}
\end{figure}

\begin{figure}[h]
    \includegraphics[scale=0.5]{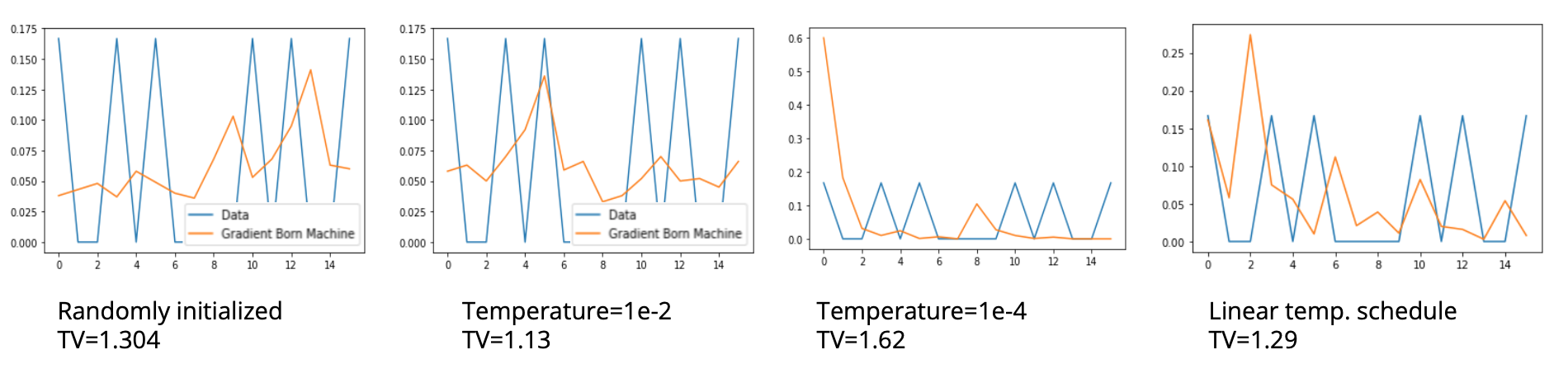}
    \caption{Classical GAN+Gumbel with various temperature settings}
    \label{fig: gumbel}
\end{figure}

\begin{figure}[h]
    \includegraphics[scale=0.4]{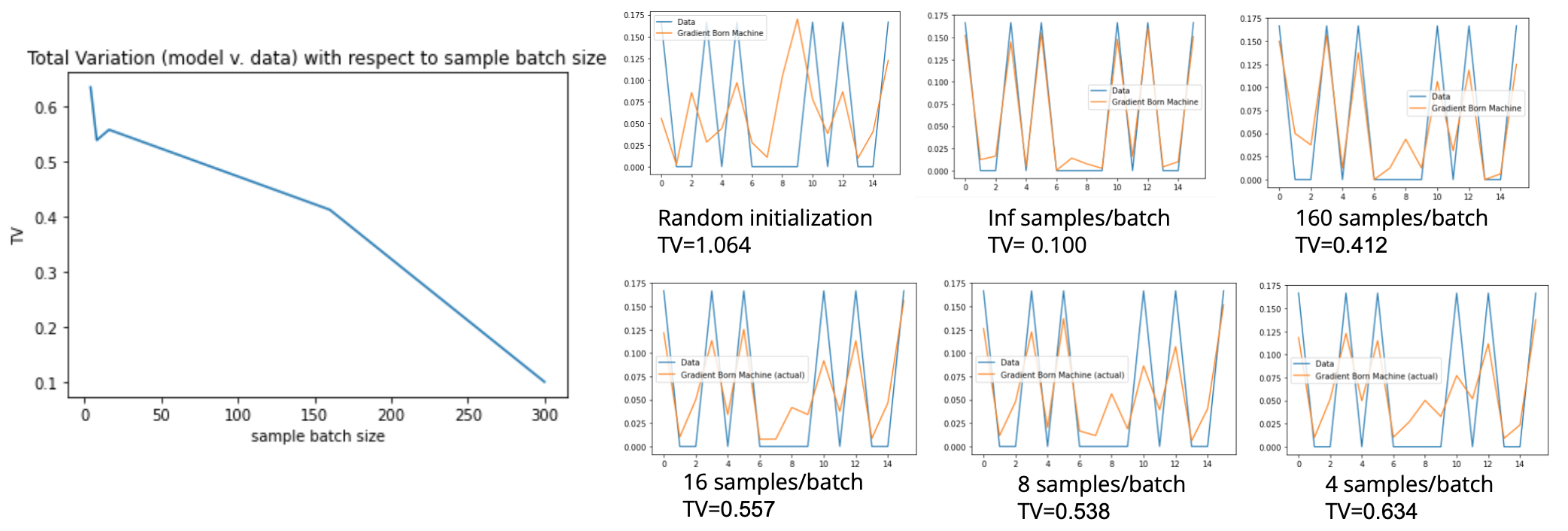}
    \caption{Born Machine+Gaussian Kernel+MMD loss with various batch sizes. Notice how the learning quality decreases drastically with smaller batches and more noises start to occur.}
    \label{fig: kernel}
\end{figure}

\begin{figure}[h]
    \includegraphics[scale=0.4]{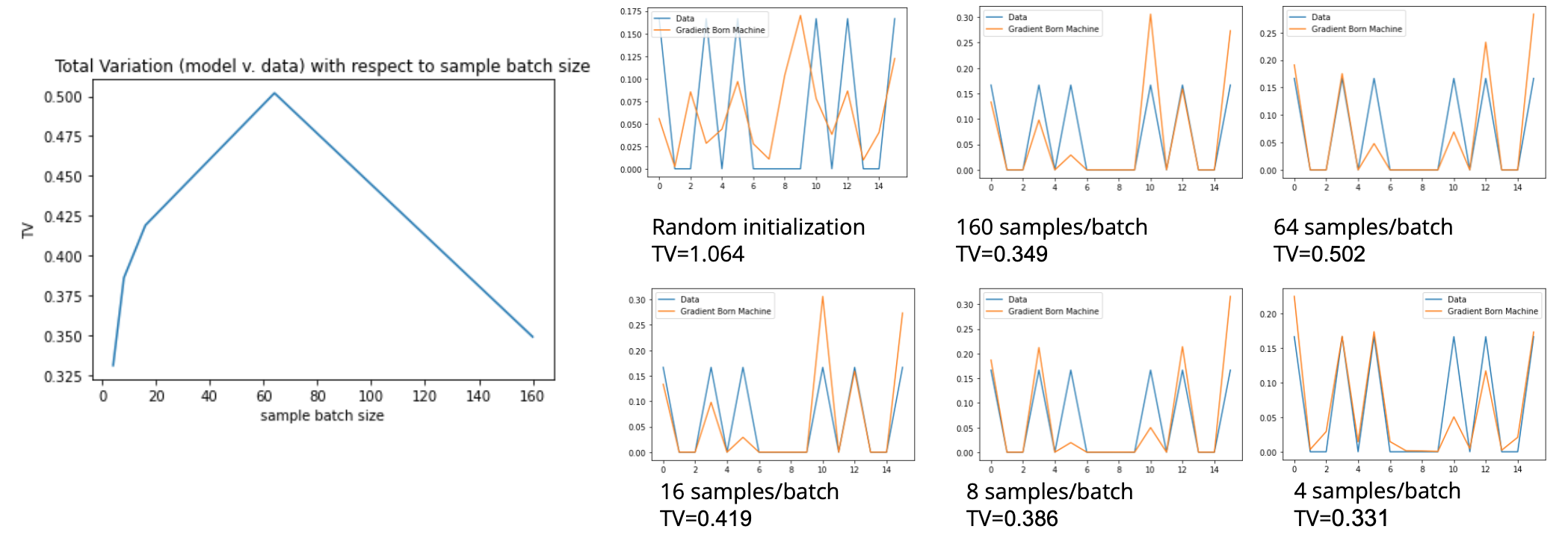}
    \caption{ Born Machine+Adversarial
Training+Non-saturating GAN loss. The learning results are not noisy, but seriously mode collapses occur at all batch size levels.}
    \label{fig: vanilla}
\end{figure}

\begin{figure}[h]
    \includegraphics[scale=0.4]{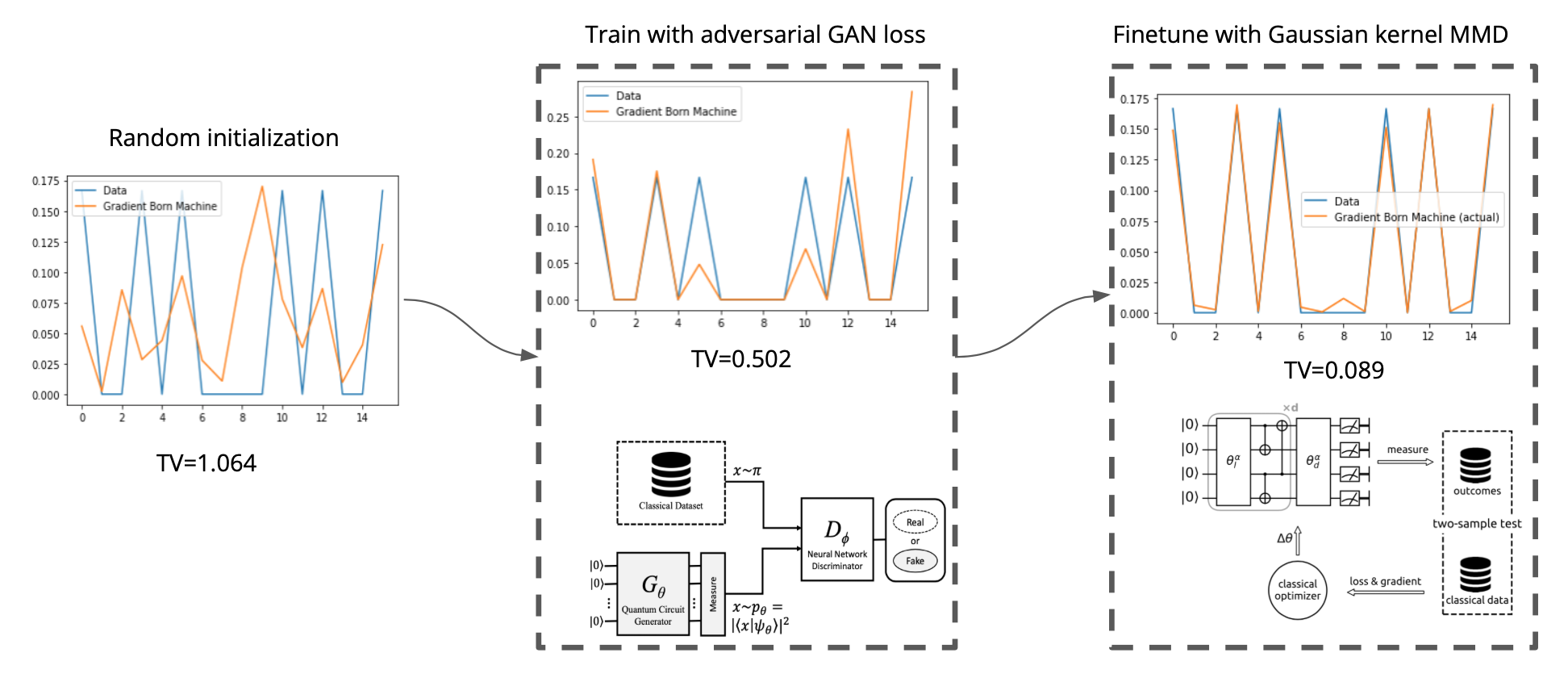}
    \caption{Fine-tuning the adversarially learnt quantum circuit with the Gaussian kernel+MMD loss can effectively fix the modal collapse issue.}
    \label{fig: finetune}
\end{figure}

\begin{figure}[h]
    \includegraphics[scale=0.5]{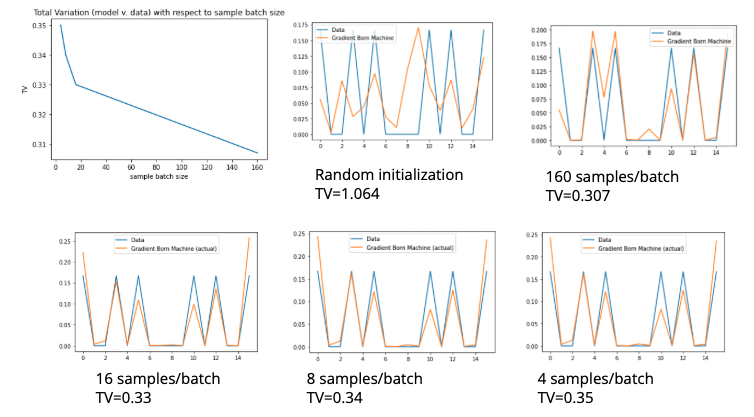}
    \caption{Born Machine+Adversarial training+MCR loss. Notice that the learning results take a middle ground between modal collapse and noisiness.}
    \label{fig: mcr}
\end{figure}

\begin{figure}[h]
    \includegraphics[scale=0.55]{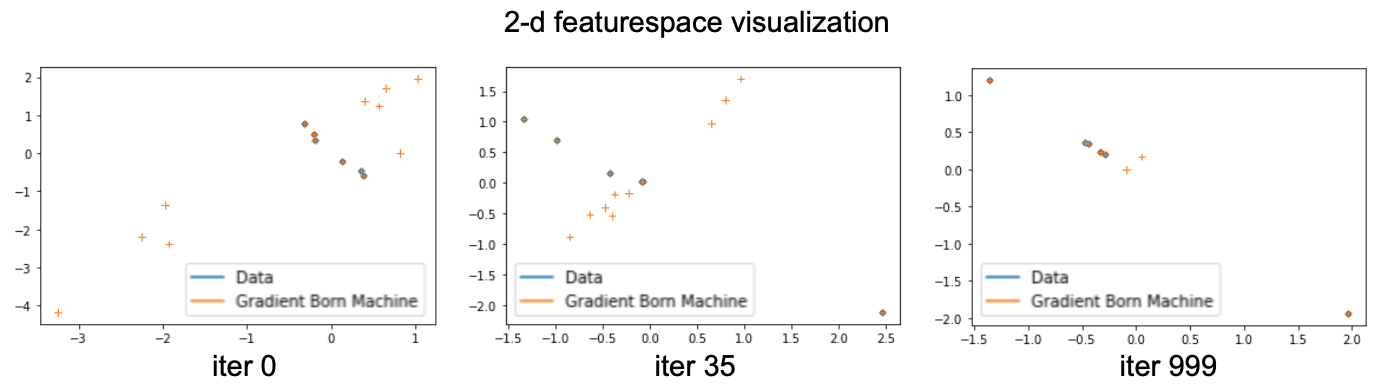}
    \caption{Visualization of the feature space. At early training stages, the Discriminator maps the real and synthesized data points into almost orthognoal linear subspaces. The quantum Generator then learns to align with the real data in later stages.}
    \label{fig: subspace}
\end{figure}

\begin{figure}[h]
    \includegraphics[scale=0.5]{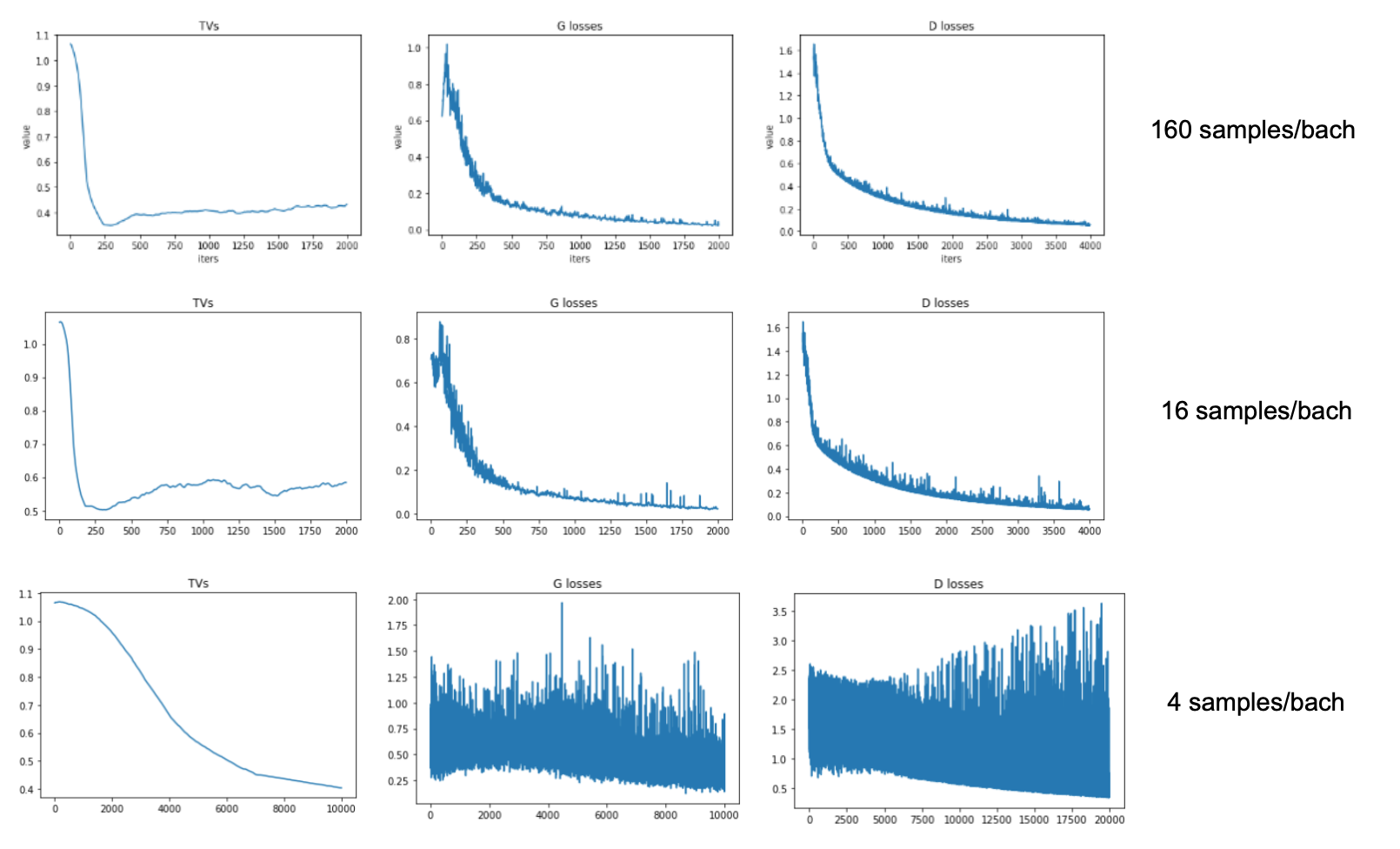}
    \caption{The G and D losses for the adversarially trained Born Machine with non-saturating GAN loss. Notice that the TV converges, yet the G and D losses experience larger fluctuations as the batch sizes decrease, which can be a potential cause of traning instability.}
    \label{fig: loss_curves}
\end{figure}

\begin{figure}[h]
    \includegraphics[scale=0.5]{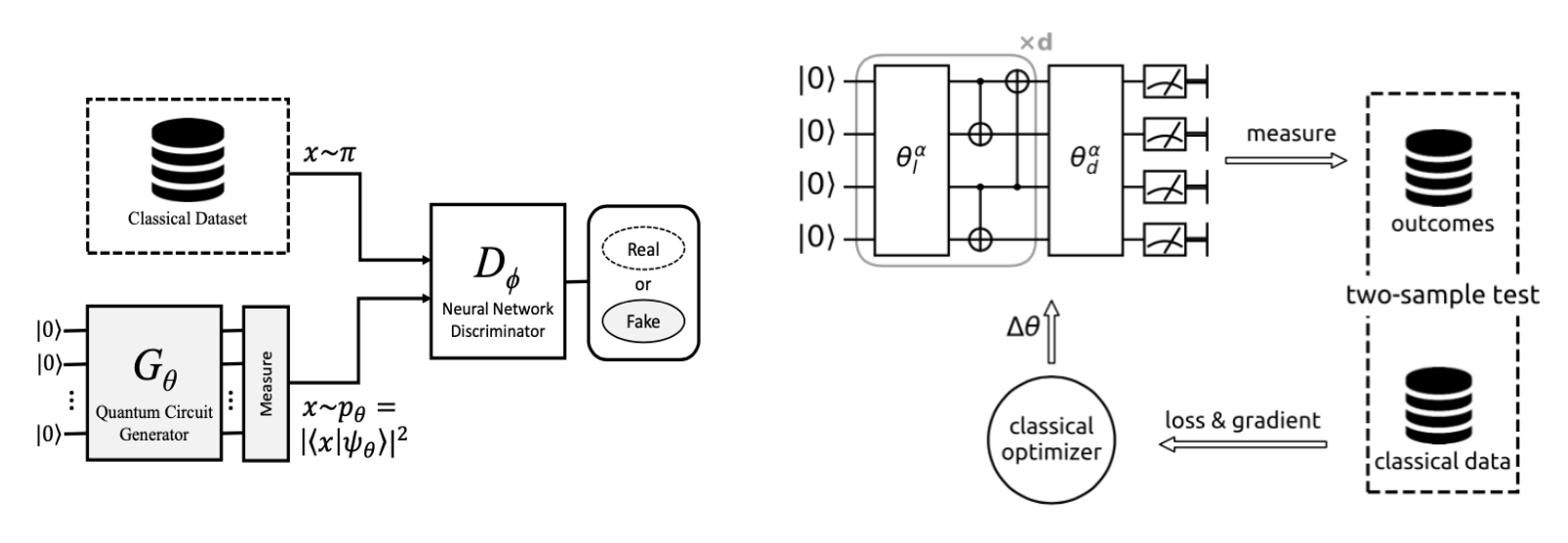}
    \caption{The Born Machine+Discriminator and Born Machine+Gaussian Kernel architectures, images taken from \cite{adv_Born} and \cite{dif_born}}
    \label{fig: architecture}
\end{figure}

\null\newpage 

\section{MCR Quantum Gradient Derivations (Neural Network Version)}
In the case when we have explicit mapping  from data space to feature space, i.e., $Z = \phi(X)\in \mathbb{R}^{n\times m}$, where $\phi(\cdot)$ can be a neural network parameterized by $\theta$. To avoid confusion, in this section I denote the feature matrix $Z=f_{\theta}(X)\in \mathbb{R}^{n\times m}$. Not requiring the inner product form now, the quantum gradient is computationally less expensive. Given Eq. (\ref{infinite}), which is replicated here:
\begin{equation}
    \begin{split}
        \Delta \lim_{m\rightarrow \infty} R([\phi(\X),\phi(\Y)])=&\underbrace{\frac{1}{2}\log\det(\I + \frac{d}{2\epsilon^2}\sum_{\x} p_{\boldsymbol{\theta}}(\x)\phi(\x)\phi(\x)^T+\frac{d}{2\epsilon^2}\sum_{\y} p_{\Y}(\y)\phi(\y)\phi(\y)^T)}_{\mathcal{L}_A}\\
        &-\underbrace{\frac{1}{4}\log\det(\I + \frac{d}{\epsilon^2}\sum_{\x} p_{\boldsymbol{\theta}}(\x)\phi(\x)\phi(\x)^T)}_{\mathcal{L}_B}-\frac{1}{4}\log\det(\I + \frac{d}{\epsilon^2}\sum_{\y} p_{\Y}(\y)\phi(\y)\phi(\y)^T)\\
        &= \mathcal{L},
    \end{split}
\end{equation}
and directly using the fact that $\frac{\partial \log\det(A)}{\partial A}=A^{-1}$, we have:
\begin{equation}
    \begin{split}
        \frac{\partial \mathcal{L}_A}{\partial p_\theta (\hat{\x})}=\frac{1}{2}<(\I + \frac{d}{2\epsilon^2}\sum_{\x} p_{\boldsymbol{\theta}}(\x)\phi(\x)\phi(\x)^T+\frac{d}{2\epsilon^2}\sum_{\y} p_{\Y}(\y)\phi(\y)\phi(\y)^T)^{-1}, \phi(\hat{\x})\phi(\hat{\x})^\top>,
    \end{split}
\end{equation}
and similarly:
\begin{equation}
    \begin{split}
        &\frac{\partial \mathcal{L}_B}{\partial p_\theta (\hat{\x})}=\frac{1}{4}<(\I + \frac{d}{\epsilon^2}\sum_{\x} p_{\boldsymbol{\theta}}(\x)\phi(\x)\phi(\x)^T)^{-1}, \phi(\hat{\x})\phi(\hat{\x})^\top>,
    \end{split}
\end{equation}
where $<A,B>$ denotes the element-wise dot product between matrices $A$ and $B$.

Therefore,
\begin{equation}
    \begin{split}
        \frac{\partial \mathcal{L}_A}{\partial \hat{\theta}}&=\sum_{\x} \frac{\partial \mathcal{L}_A}{\partial p_{\boldsymbol{\theta}}(\x)}\frac{\partial p_{\boldsymbol{\theta}}(\x)}{\partial \hat{\theta}}=\sum_{\x} \frac{\partial \mathcal{L}_A}{\partial p_{\boldsymbol{\theta}}(\x)}\frac{p_{\theta^+}(\x)-p_{\theta^-}(\x)}{2}\\
        &=\frac{1}{4}\mathbb{E}_{\hat{\x}\sim p_{\theta^+}}<(\I + \frac{d}{2\epsilon^2}\sum_{\x} p_{\boldsymbol{\theta}}(\x)\phi(\x)\phi(\x)^T+\frac{d}{2\epsilon^2}\sum_{\y} p_{\Y}(\y)\phi(\y)\phi(\y)^T)^{-1}, \phi(\hat{\x})\phi(\hat{\x})^\top>\\
        &-\frac{1}{4}\mathbb{E}_{\hat{\x}\sim p_{\theta^-}}<(\I + \frac{d}{2\epsilon^2}\sum_{\x} p_{\boldsymbol{\theta}}(\x)\phi(\x)\phi(\x)^T+\frac{d}{2\epsilon^2}\sum_{\y} p_{\Y}(\y)\phi(\y)\phi(\y)^T)^{-1}, \phi(\hat{\x})\phi(\hat{\x})^\top>,
    \end{split}
\end{equation}
and
\begin{equation}
    \begin{split}
        \frac{\partial \mathcal{L}_B}{\partial \hat{\theta}}&=\sum_{\x} \frac{\partial \mathcal{L}_B}{\partial p_{\boldsymbol{\theta}}(\x)}\frac{\partial p_{\boldsymbol{\theta}}(\x)}{\partial \hat{\theta}}=\sum_{\x} \frac{\partial \mathcal{L}_B}{\partial p_{\boldsymbol{\theta}}(\x)}\frac{p_{\theta^+}(\x)-p_{\theta^-}(\x)}{2}\\
        &=\frac{1}{8}\mathbb{E}_{\hat{\x}\sim p_{\theta^+}}<(\I + \frac{d}{\epsilon^2}\sum_{\x} p_{\boldsymbol{\theta}}(\x)\phi(\x)\phi(\x)^T)^{-1}, \phi(\hat{\x})\phi(\hat{\x})^\top>\\
        &-\frac{1}{8}\mathbb{E}_{\hat{\x}\sim p_{\theta^-}}<(\I + \frac{d}{\epsilon^2}\sum_{\x} p_{\boldsymbol{\theta}}(\x)\phi(\x)\phi(\x)^T)^{-1}, \phi(\hat{\x})\phi(\hat{\x})^\top>,
    \end{split}
\end{equation}

\section{MCR Quantum Gradient Derivations (Kernel Version)}
In the simplest case (discrete and compact distributions $ p_X$ and $ p_Y$), given two groups of samples, $x\sim p_X$ and $y\sim p_Y$, and if we treat each sample as a column vector and construct $\X$ and $\Y$, both are $\mathcal{R}^{d\times m}$, where $d$ is the dimension of features and $m$ is the batch size, we define the MCR (second order) ``distance'' between $\X$ and $\Y$ as:
\begin{equation}
\label{sample}
    \begin{split}
        \Delta R([\X,\Y])=&\frac{1}{2}\log\det(\I + \frac{d}{2m\epsilon^2}\X\X^T+\frac{d}{2m\epsilon^2}\Y\Y^T)\\
        &-\frac{1}{4}\log\det(\I + \frac{d}{m\epsilon^2}\X\X^T)-\frac{1}{4}\log\det(\I + \frac{d}{m\epsilon^2}\Y\Y^T).
    \end{split}
\end{equation}
This metric is $0$ iff $\frac{1}{m}\X\X^T=\frac{1}{m}\Y\Y^T$, i.e., the sample covariance matrices (given the same batch sizes) are equal. 

For the case when the sample size $m\rightarrow \infty$, we have:
\begin{equation}
    \lim_{m\rightarrow \infty} \frac{d}{m\epsilon^2}\X\X^T=\frac{d}{\epsilon^2}\sum_{\x} p_{\X}(\x)\x\x^T.
\end{equation}

Therefore:
\begin{equation}
\label{infinite}
    \begin{split}
        \Delta \lim_{m\rightarrow \infty} R([\X,\Y])=&\frac{1}{2}\log\det(\I + \frac{d}{2\epsilon^2}\sum_{\x} p_{\X}(\x)\x\x^T+\frac{d}{2\epsilon^2}\sum_{\y} p_{\Y}(\y)\y\y^T)\\
        &-\frac{1}{4}\log\det(\I + \frac{d}{\epsilon^2}\sum_{\x} p_{\X}(\x)\x\x^T)-\frac{1}{4}\log\det(\I + \frac{d}{\epsilon^2}\sum_{\y} p_{\Y}(\y)\y\y^T).
    \end{split}
\end{equation}
In the case when $p_X$ and $p_Y$ are Gaussian distributions, matching the above objective (covariances) along with the first moment (mean) will be sufficient for showing $p_X=p_Y$. In practice, $p_X$ and $p_Y$ are intractable, so we can only approximate (\ref{infinite}) by (\ref{sample}) with large enough $m$.
\\

We can also compare the covariance of features in a RHKS kernel space via a feature mapping $\phi(\cdot)$. By the law of unconscious statistician, $\mathbb{E}_{\phi(\x)\sim p_{\phi(x)}}\phi(\x)\phi(\x)^T=\mathbb{E}_{\x\sim p_{\X}}\phi(\x)\phi(\x)^T$, so the covariance matrix in kernel feature space can be written as the expectation over $p_{\X}$:
\begin{equation}
\label{kernel_}
    \begin{split}
        \Delta \lim_{m\rightarrow \infty} R([\phi(\X),\phi(\Y)])=&\frac{1}{2}\log\det(\I + \frac{d}{2\epsilon^2}\sum_{\x} p_{\X}(\x)\phi(\x)\phi(\x)^T+\frac{d}{2\epsilon^2}\sum_{\y} p_{\Y}(\y)\phi(\y)\phi(\y)^T)\\
        &-\frac{1}{4}\log\det(\I + \frac{d}{\epsilon^2}\sum_{\x} p_{\X}(\x)\phi(\x)\phi(\x)^T)-\frac{1}{4}\log\det(\I + \frac{d}{\epsilon^2}\sum_{\y} p_{\Y}(\y)\phi(\y)\phi(\y)^T).
    \end{split}
\end{equation}
\\

Designating $p_{\X}$ to be the distribution that we want to parameterize, the variational quantum circuit uses a vector of parameters $\boldsymbol{\theta}$ to model $p_{\X}$:
\begin{equation}
\label{kernel_}
    \begin{split}
        \Delta \lim_{m\rightarrow \infty} R([\phi(\X),\phi(\Y)])=&\underbrace{\frac{1}{2}\log\det(\I + \frac{d}{2\epsilon^2}\sum_{\x} p_{\boldsymbol{\theta}}(\x)\phi(\x)\phi(\x)^T+\frac{d}{2\epsilon^2}\sum_{\y} p_{\Y}(\y)\phi(\y)\phi(\y)^T)}_{\mathcal{L}_A}\\
        &-\underbrace{\frac{1}{4}\log\det(\I + \frac{d}{\epsilon^2}\sum_{\x} p_{\boldsymbol{\theta}}(\x)\phi(\x)\phi(\x)^T)}_{\mathcal{L}_B}-\frac{1}{4}\log\det(\I + \frac{d}{\epsilon^2}\sum_{\y} p_{\Y}(\y)\phi(\y)\phi(\y)^T)\\
        &= \mathcal{L}
    \end{split}
\end{equation}
\\

Now we derive the quantum gradient of the above metric with respect to a specific single quantum circuit parameter $\hat{\theta}\in \boldsymbol{\theta}$. First note, by the structure of the above loss function:
\begin{equation}
    \frac{\partial \mathcal{L}}{\partial \hat{\theta}}=\sum_{\x} \frac{\partial \mathcal{L}}{\partial p_{\boldsymbol{\theta}}(\x)}\frac{\partial p_{\boldsymbol{\theta}}(\x)}{\partial \hat{\theta}},
\end{equation}
where 
\begin{equation}
    \frac{\partial p_{\boldsymbol{\theta}}(\x)}{\partial \hat{\theta}}=\frac{1}{2}\left(p_{\theta^{+}}(\x)-p_{\theta^{-}}(\x)\right).
\end{equation}

There are \textbf{two ways} of expressing $\frac{\partial \mathcal{L}}{\partial p_{\boldsymbol{\theta}}(\x)}$: one taking less computation, but explicitly needs the expression for $\phi(\cdot)$, and one taking more computation, but can be written in \textbf{dot product form} in the kernel space.
\\

To be able to use universal kernels, I first derive the \textbf{second formulation}, using the fact that:
\begin{equation}
    \begin{split}
      \frac{d}{d x} \log \operatorname{det} A(x)= \frac{d}{d x} \operatorname{Tr} \log A(x)= \operatorname{Tr}\left(A(x)^{-1} \frac{d A}{d x}\right),
    \end{split}
\end{equation}
regardless whether matrices $\frac{d A}{d x}$ and $A(x)$ commute or not (add derivation to appendix).
\\

\begin{equation}
    \begin{split}
      2\frac{\partial \mathcal{L}_A}{\partial p_{\boldsymbol{\theta}}(\hat{\x})}&=\Tr \Big[\Big(\I + \frac{d}{2\epsilon^2}\sum_{\x} p_{\boldsymbol{\theta}}(\x)\phi(\x)\phi(\x)^T+\frac{d}{2\epsilon^2}\sum_{\y} p_{\Y}(\y)\phi(\y)\phi(\y)^T \Big)^{-1} \phi(\hat{\x})\phi(\hat{\x})^T \Big]\\
      &=\lim_{m\rightarrow \infty} \Tr \Big[\Big(\I + \frac{d}{2m\epsilon^2}\phi(\X)\phi(\X)^T+\frac{d}{2m\epsilon^2}\phi(\Y)\phi(\Y)^T \Big)^{-1} \phi(\hat{\x})\phi(\hat{\x})^T \Big]\\
      &=\lim_{m\rightarrow \infty} \Tr \Big[\Big(\I + \frac{d}{2m\epsilon^2}\underbrace{[\phi(\X), \phi(\Y)]}_{M}\underbrace{[\phi(\X), \phi(\Y)]^T}_{M^T} \Big)^{-1} \phi(\hat{\x})\phi(\hat{\x})^T \Big]\\
      &=\lim_{m\rightarrow \infty} \Tr \Big[\Big(\I - \frac{d}{2m\epsilon^2}M\Big(\I+\frac{d}{2m\epsilon^2} M^TM\Big)^{-1} M^T \Big) \phi(\hat{\x})\phi(\hat{\x})^T \Big]\\
      &=\lim_{m\rightarrow \infty} \Tr \Big[\phi(\hat{\x})\phi(\hat{\x})^T - \frac{d}{2m\epsilon^2}M\Big(\I+\frac{d}{2m\epsilon^2} M^TM\Big)^{-1} M^T  \phi(\hat{\x})\phi(\hat{\x})^T \Big]\\
      &=\lim_{m\rightarrow \infty} \Tr \Big[\phi(\hat{\x})\phi(\hat{\x})^T\Big] - \Tr\Big[ \frac{d}{2m\epsilon^2}M\Big(\I+\frac{d}{2m\epsilon^2} M^TM\Big)^{-1} M^T  \phi(\hat{\x})\phi(\hat{\x})^T \Big]\\
      &=\lim_{m\rightarrow \infty} \Tr \Big[\phi(\hat{\x})^T \phi(\hat{\x})\Big] - \Tr\Big[ \frac{d}{2m\epsilon^2}\phi(\hat{\x})^T M\Big(\I+\frac{d}{2m\epsilon^2} M^TM\Big)^{-1} M^T  \phi(\hat{\x}) \Big]\\
      &=\lim_{m\rightarrow \infty} \Big\{\phi(\hat{\x})^T \phi(\hat{\x}) - \frac{d}{2m\epsilon^2}\phi(\hat{\x})^T M\Big(\I+\frac{d}{2m\epsilon^2} M^TM\Big)^{-1} M^T  \phi(\hat{\x}) \Big\}
    \end{split}
\end{equation}
Then, the gradient of $\mathcal{L}_A$ is expressed as:
\begin{equation}
    \begin{split}
        2\frac{\partial \mathcal{L}_A}{\partial \hat{\theta}}&=\sum_{\x} \frac{\partial \mathcal{L}}{\partial p_{\boldsymbol{\theta}}(\x)}\frac{\partial p_{\boldsymbol{\theta}}(\x)}{\partial \hat{\theta}}=\lim_{m\rightarrow \infty}\sum_{\hat{\x}}  \Big[\phi(\hat{\x})^T \phi(\hat{\x})\Big]\frac{\partial p_{\boldsymbol{\theta}}(\x)}{\partial \hat{\theta}}\\
        &-\sum_{\hat{\x}} \Big[ \frac{d}{2m\epsilon^2}\phi(\hat{\x})^T M\Big(\I+\frac{d}{2m\epsilon^2} M^TM\Big)^{-1} M^T  \phi(\hat{\x}) \Big]\frac{\partial p_{\boldsymbol{\theta}}(\x)}{\partial \hat{\theta}}\\
        &=\frac{1}{2}\mathbb{E}_{\x\sim p_{\boldsymbol{\theta^+}}}\phi(\hat{\x})^T \phi(\hat{\x})-\frac{1}{2}\mathbb{E}_{\x \sim p_{\boldsymbol{\theta^-}}}\phi(\hat{\x})^T \phi(\hat{\x})\\
        &-\lim_{m\rightarrow \infty}\frac{1}{2}\Bigg\{\frac{d}{2m\epsilon^2}\mathbb{E}_{\x \sim p_{\boldsymbol{\theta^+}}}\phi(\hat{\x})^T M\Big(\I+\frac{d}{2m\epsilon^2} M^TM\Big)^{-1} M^T  \phi(\hat{\x}) \\
        &- \frac{d}{2m\epsilon^2}\mathbb{E}_{\x \sim p_{\boldsymbol{\theta^-}}}\phi(\hat{\x})^T M\Big(\I+\frac{d}{2m\epsilon^2} M^TM\Big)^{-1} M^T  \phi(\hat{\x}) \Bigg\}.\\
    \end{split}
\end{equation}
Using the similar procedure, we can see the gradient of $\mathcal{L}_B$ is:
\begin{equation}
    \begin{split}
        4\frac{\partial \mathcal{L}_B}{\partial \hat{\theta}}&=\frac{1}{2}\mathbb{E}_{\x\sim p_{\boldsymbol{\theta^+}}}\phi(\hat{\x})^T \phi(\hat{\x})-\frac{1}{2}\mathbb{E}_{\x \sim p_{\boldsymbol{\theta^-}}}\phi(\hat{\x})^T \phi(\hat{\x})\\
        &-\lim_{m\rightarrow \infty}\frac{1}{2}\Bigg\{\frac{d}{m\epsilon^2}\mathbb{E}_{\x \sim p_{\boldsymbol{\theta^+}}}\phi(\hat{\x})^T \phi(\X)\Big(\I+\frac{d}{m\epsilon^2} \phi(\X)^T \phi(\X)\Big)^{-1} \phi(\X)^T  \phi(\hat{\x}) \\
        &- \frac{d}{m\epsilon^2}\mathbb{E}_{\x \sim p_{\boldsymbol{\theta^-}}}\phi(\hat{\x})^T \phi(\X)\Big(\I+\frac{d}{m\epsilon^2} \phi(\X)^T \phi(\X)\Big)^{-1} \phi(\X)^T  \phi(\hat{\x}) \Bigg\}.\\
    \end{split}
\end{equation}

\bibliographystyle{plain}
\bibliography{references}

\begin{thebibliography}{10}

\bibitem{aaronson_quantum_2005}
Scott Aaronson.
\newblock Quantum computing, postselection, and probabilistic polynomial-time.
\newblock {\em Proceedings of the Royal Society A: Mathematical, Physical and
  Engineering Sciences}, 461(2063):3473--3482, November 2005.

\bibitem{gumbel_clustering}
Deepak~Bhaskar Acharya and Huaming Zhang.
\newblock Community detection clustering via gumbel softmax.
\newblock {\em SN Computer Science}, 1(5):1--11, 2020.

\bibitem{boaz_complexity}
Sanjeev Arora and Boaz Barak.
\newblock {\em Computational Complexity: A Modern Approach}.
\newblock Cambridge University Press, USA, 1st edition, 2009.

\bibitem{gumbel_speech}
Alexei Baevski, Steffen Schneider, and Michael Auli.
\newblock vq-wav2vec: Self-supervised learning of discrete speech
  representations, 2019.

\bibitem{benedetti-quantum-shallow}
Marcello Benedetti, Delfina Garcia-Pintos, Oscar Perdomo, Vicente
  Leyton-Ortega, Yunseong Nam, and Alejandro Perdomo-Ortiz.
\newblock A generative modeling approach for benchmarking and training shallow
  quantum circuits.
\newblock {\em npj Quantum Information}, 5(1), May 2019.

\bibitem{binkowski2018demystifying}
Miko{\l}aj Bi{\'n}kowski, Dougal~J Sutherland, Michael Arbel, and Arthur
  Gretton.
\newblock Demystifying mmd gans.
\newblock {\em arXiv preprint arXiv:1801.01401}, 2018.

\bibitem{boixo_supremacy}
Sergio Boixo, Sergei~V. Isakov, Vadim~N. Smelyanskiy, Ryan Babbush, Nan Ding,
  Zhang Jiang, Michael~J. Bremner, John~M. Martinis, and Hartmut Neven.
\newblock Characterizing quantum supremacy in near-term devices.
\newblock {\em Nature Physics}, 14(6):595--600, apr 2018.

\bibitem{Born_rule_german}
Max Born.
\newblock Zur quantenmechanik der sto{\ss}vorg{\"a}nge.
\newblock {\em Zeitschrift f{\"u}r Physik}, 37(12):863--867, 1926.

\bibitem{Bremner_PH}
Michael~J. Bremner, Richard Jozsa, and Dan~J. Shepherd.
\newblock Classical simulation of commuting quantum computations implies
  collapse of the polynomial hierarchy.
\newblock {\em Proceedings of the Royal Society A: Mathematical, Physical and
  Engineering Sciences}, 467(2126):459--472, aug 2010.

\bibitem{chan2020redu}
Kwan Ho~Ryan Chan, Yaodong Yu, Chong You, Haozhi Qi, John Wright, and Yi~Ma.
\newblock Deep networks from the principle of rate reduction, 2020.

\bibitem{ReduNet}
Kwan Ho~Ryan Chan, Yaodong Yu, Chong You, Haozhi Qi, John Wright, and Yi~Ma.
\newblock Redunet: {A} white-box deep network from the principle of maximizing
  rate reduction.
\newblock {\em CoRR}, abs/2105.10446, 2021.

\bibitem{che2016mode}
Tong Che, Yanran Li, Athul~Paul Jacob, Yoshua Bengio, and Wenjie Li.
\newblock Mode regularized generative adversarial networks.
\newblock {\em arXiv preprint arXiv:1612.02136}, 2016.

\bibitem{infogan}
Xi~Chen, Yan Duan, Rein Houthooft, John Schulman, Ilya Sutskever, and Pieter
  Abbeel.
\newblock {InfoGAN}: Interpretable representation learning by information
  maximizing generative adversarial nets.
\newblock In {\em Advances in Neural Information Processing Systems}, pages
  2172--2180, 2016.

\bibitem{ref15dif}
Song Cheng, Jing Chen, and Lei Wang.
\newblock Information perspective to probabilistic modeling: Boltzmann machines
  versus born machines.
\newblock {\em Entropy}, 20(8):583, Aug 2018.

\bibitem{born_boltzmann}
Song Cheng, Jing Chen, and Lei Wang.
\newblock Information perspective to probabilistic modeling: Boltzmann machines
  versus born machines.
\newblock {\em Entropy}, 20(8):583, aug 2018.

\bibitem{expressive_q}
Yuxuan Du, Min-Hsiu Hsieh, Tongliang Liu, and Dacheng Tao.
\newblock Expressive power of parametrized quantum circuits.
\newblock {\em Physical Review Research}, 2(3), jul 2020.

\bibitem{Nash-Gan}
Farzan Farnia and Asuman~E. Ozdaglar.
\newblock Gans may have no nash equilibria.
\newblock {\em CoRR}, abs/2002.09124, 2020.

\bibitem{Fedus2018MaskGANBT}
William Fedus, Ian~J. Goodfellow, and Andrew~M. Dai.
\newblock Maskgan: Better text generation via filling in the \_\_\_\_\_\_.
\newblock {\em ArXiv}, abs/1801.07736, 2018.

\bibitem{GAN-LOG}
Soheil Feizi, Farzan Farnia, Tony Ginart, and David Tse.
\newblock Understanding gans in the lqg setting: Formulation, generalization
  and stability.
\newblock {\em IEEE Journal on Selected Areas in Information Theory},
  1(1):304--311, 2020.

\bibitem{gao2021maximum}
Ruize Gao, Feng Liu, Jingfeng Zhang, Bo~Han, Tongliang Liu, Gang Niu, and
  Masashi Sugiyama.
\newblock Maximum mean discrepancy test is aware of adversarial attacks, 2021.

\bibitem{goodfellow2014generative}
Ian Goodfellow, Jean Pouget-Abadie, Mehdi Mirza, Bing Xu, David Warde-Farley,
  Sherjil Ozair, Aaron Courville, and Yoshua Bengio.
\newblock Generative adversarial nets.
\newblock In {\em Advances in neural information processing systems}, pages
  2672--2680, 2014.

\bibitem{MMD_gretton}
Arthur Gretton, Karsten Borgwardt, Malte~J. Rasch, Bernhard Scholkopf, and
  Alexander~J. Smola.
\newblock A kernel method for the two-sample problem, 2008.

\bibitem{Gumbel_original}
Emil~Julius Gumbel.
\newblock {\em Statistical theory of extreme values and some practical
  applications; a series of lectures}.
\newblock Applied mathematics series ; 33. U.S. Govt. Print. Office,
  Washington, 1954.

\bibitem{ref14dif}
Zhao-Yu Han, Jun Wang, Heng Fan, Lei Wang, and Pan Zhang.
\newblock Unsupervised generative modeling using matrix product states.
\newblock {\em Physical Review X}, 8(3), Jul 2018.

\bibitem{born_rule}
Zhao-Yu Han, Jun Wang, Heng Fan, Lei Wang, and Pan Zhang.
\newblock Unsupervised generative modeling using matrix product states.
\newblock {\em Phys. Rev. X}, 8:031012, Jul 2018.

\bibitem{kernel_space}
Thomas Hofmann, Bernhard Schölkopf, and Alexander~J. Smola.
\newblock Kernel methods in machine learning.
\newblock {\em The Annals of Statistics}, 36(3):1171--1220, 2008.

\bibitem{gumbel_softmax}
Eric Jang, Shixiang Gu, and Ben Poole.
\newblock Categorical reparameterization with gumbel-softmax, 2016.

\bibitem{Kramer1991NonlinearPC}
Mark~A. Kramer.
\newblock Nonlinear principal component analysis using autoassociative neural
  networks.
\newblock {\em Aiche Journal}, 37:233--243, 1991.

\bibitem{Gumbel_GAN}
Matt~J. Kusner and José~Miguel Hernández-Lobato.
\newblock Gans for sequences of discrete elements with the gumbel-softmax
  distribution, 2016.

\bibitem{lala2018evaluation}
Sayeri Lala, Maha Shady, Anastasiya Belyaeva, and Molei Liu.
\newblock Evaluation of mode collapse in generative adversarial networks.

\bibitem{VAE-GAN}
Anders Boesen~Lindbo Larsen, S{\o}ren~Kaae S{\o}nderby, Hugo Larochelle, and
  Ole Winther.
\newblock Autoencoding beyond pixels using a learned similarity metric.
\newblock {\em arXiv preprint arXiv:1512.09300}, 2015.

\bibitem{quantum_gradient_36}
Jun Li, Xiaodong Yang, Xinhua Peng, and Chang-Pu Sun.
\newblock Hybrid quantum-classical approach to quantum optimal control.
\newblock {\em Phys. Rev. Lett.}, 118:150503, Apr 2017.

\bibitem{liu2021learning}
Feng Liu, Wenkai Xu, Jie Lu, Guangquan Zhang, Arthur Gretton, and Danica~J.
  Sutherland.
\newblock Learning deep kernels for non-parametric two-sample tests, 2021.

\bibitem{dif_born}
Jin-Guo Liu and Lei Wang.
\newblock Differentiable learning of quantum circuit born machines.
\newblock {\em Physical Review A}, 98(6), Dec 2018.

\bibitem{cgan}
Mehdi Mirza and Simon Osindero.
\newblock Conditional generative adversarial nets.
\newblock {\em arXiv preprint arXiv:1411.1784}, 2014.

\bibitem{quantum_gradient_38}
K.~Mitarai, M.~Negoro, M.~Kitagawa, and K.~Fujii.
\newblock Quantum circuit learning.
\newblock {\em Physical Review A}, 98(3), sep 2018.

\bibitem{BAS}
Niels Neumann, Frank Phillipson, and Richard Versluis.
\newblock Machine learning in the quantum era.
\newblock {\em Digitale Welt}, 3:24--29, 04 2019.

\bibitem{veegan_mode_collapse}
Akash Srivastava, Lazar Valkov, Chris Russell, Michael~U. Gutmann, and Charles
  Sutton.
\newblock Veegan: Reducing mode collapse in gans using implicit variational
  learning, 2017.

\bibitem{Stein2020QuGANAG}
Samuel~A. Stein, Betis Baheri, Daniel Chen, Ying Mao, Qiang Guan, Ang Li,
  Bo~Fang, and Shuai Xu.
\newblock Qugan: A generative adversarial network through quantum states.
\newblock {\em arXiv: Quantum Physics}, 2020.

\bibitem{eecs_BPP}
Vazirani.
\newblock Bqp vs. p, bpp, accuracy.
\newblock 2004.

\bibitem{kernel_primer}
J.P. Vert, Koji Tsuda, and Bernhard SchÃ¶lkopf.
\newblock A primer on kernel methods.
\newblock {\em Kernel Methods in Computational Biology, 35-70 (2004)}, 01 2004.

\bibitem{wenliang2021learning}
Li~Wenliang, Danica~J. Sutherland, Heiko Strathmann, and Arthur Gretton.
\newblock Learning deep kernels for exponential family densities, 2021.

\bibitem{seqGAN}
Lantao Yu, Weinan Zhang, Jun Wang, and Yong Yu.
\newblock Seqgan: Sequence generative adversarial nets with policy gradient.
\newblock {\em CoRR}, abs/1609.05473, 2016.

\bibitem{yu2020learning}
Yaodong Yu, Kwan Ho~Ryan Chan, Chong You, Chaobing Song, and Yi~Ma.
\newblock Learning diverse and discriminative representations via the principle
  of maximal coding rate reduction.
\newblock In {\em Advances in neural information processing systems}, 2020.

\bibitem{adv_Born}
Jinfeng Zeng, Yufeng Wu, Jin-Guo Liu, Lei Wang, and Jiangping Hu.
\newblock Learning and inference on generative adversarial quantum circuits.
\newblock {\em Physical Review A}, 99(5), may 2019.

\end{thebibliography}

\end{document}